\documentclass[12pt]{article}

\usepackage{scicite}
\usepackage{times}

\usepackage{amsmath}
\usepackage{amssymb}
\usepackage{txfonts}
\usepackage{mathtools}
\usepackage{quantum}
\usepackage{color}
\usepackage[all]{xy}
\usepackage{caption}
\usepackage{subcaption}
\usepackage{array}
\usepackage{enumerate}

\usepackage{bibentry} 

\usepackage[colorlinks,linktocpage,hypertexnames=false,pdfpagelabels]{hyperref}
\usepackage[amsmath,thmmarks,hyperref]{ntheorem}
\usepackage[capitalise]{cleveref}

\newtheorem{theorem}{Theorem}
\newtheorem{lemma}[theorem]{Lemma}
\newtheorem{corollary}[theorem]{Corollary}

\newtheorem{definition}[theorem]{Definition}

\theorembodyfont{\normalfont}
\newtheorem{remark}[theorem]{Remark}

\theorembodyfont{\normalfont}
\newtheorem{example}[theorem]{Example}

\theoremstyle{nonumberplain}
\theorembodyfont{\normalfont}
\theoremsymbol{$\Box$}
\newtheorem{proof}{Proof}

\theoremstyle{nonumberbreak}
\theorembodyfont{\normalfont}
\theoremsymbol{$\Box$}

\makeatletter
\renewcommand{\thefigure}{\@arabic\c@figure}
\makeatother

\makeatletter
\renewcommand{\theequation}{\@arabic\c@equation}
\makeatother

\crefname{part}{Part}{Parts}
\crefname{step}{Step}{Steps}

\renewcommand{\emptyset}{\varnothing}
\newcommand{\xnor}{\odot}

\DeclareMathOperator{\poly}{poly}

\newenvironment{sciabstract}{%
\begin{quote} \bf}
{\end{quote}}

\newcounter{lastnote}

\begin{document}

\title{Simple universal models capture all classical spin physics}

\author{Gemma De las Cuevas$^{\ast 1}$ and Toby S. Cubitt$^{2}$}
\date{}

 \maketitle

\vspace{-1cm}

\begin{center}
{\small $^1$Max Planck Institute for Quantum Optics, Hans-Kopfermann-Str.\ 1, D-85748 Garching, Germany\\
$^\ast$To whom correspondence should be addressed. E-mail: gemma.delascuevas@mpq.mpg.de\\
$^2$Department of Computer Science, University College London, 
  Gower Street, London WC1E 6EA, UK}
\end{center}

 \begin{sciabstract}
   Spin models are used in many studies of complex systems---be it condensed matter physics, neural networks, or economics---as they exhibit rich macroscopic behaviour despite their microscopic simplicity.
   Here we prove that all the physics of \emph{every} classical spin model is reproduced in the low-energy sector of certain `universal models'.
   This means that (i)~the low energy spectrum of the universal model reproduces the entire spectrum of the original model to any desired precision, (ii)~the corresponding spin configurations of the original model are also reproduced in the universal model, (iii)~the partition function is approximated to any desired precision, and (iv)~the overhead in terms of number of spins and interactions is at most polynomial.
   This holds for classical models with discrete or continuous degrees of freedom.
   We prove necessary and sufficient conditions for a spin model to be universal, and show that one of the simplest and most widely studied spin models, the 2D Ising model with fields, is universal.
 \end{sciabstract}

The description of systems with many interacting degrees of freedom is a ubiquitous problem across the natural and social sciences.
Be it electrons in a material, neurons interacting through synapses, or speculative agents in a market, the challenge is to simplify the system so that it becomes tractable while capturing some of the relevant features of the real system.
Spin models are one way of addressing this challenge.
While originally introduced in condensed matter physics to study magnetic materials \cite{Di04b,Bi05,Ni11,Ba82}, they have now permeated many other disciplines, including quantum gravity \cite{Am09a}, error-correcting codes~\cite{Ni01}, percolation theory~\cite{Ni11}, graph theory~\cite{Bo98}, neural networks~\cite{Ro96}, protein folding ~\cite{Br91b}, and trading models in stock markets~\cite{Du99b}.

The reason for this success is that spin models are microscopically simple, yet their versatile interactions lead to a very wide variety of macroscopic behaviour.
Formally, a spin model is specified by a set of degrees of freedom, the ``spins'', and a cost function, or ``Hamiltonian'', $H$ which specifies the interaction pattern as well as the type and strength of interactions among the spins.
(In physics, the Hamiltonian specifies the energy of each possible spin configuration; in other contexts, this energy value may quantify a more abstract ``cost'' associated with a configuration.)

Naturally, this definition encompasses a wide range of models, including, e.g., attractive and/or repulsive interactions, regular and irregular interaction patterns, models in different spatial dimensions, with different symmetries (e.g.~``conventional'' spin models with global symmetries, versus models with local symmetries such as lattice gauge theories), many-body interactions (e.g.\ vertex models and edge models~\cite{De13b}), and more.
We will use the word ``model'' to refer to a (generally infinite) family of spin Hamiltonians.
Different Hamiltonians within the same model are typically related in some natural way.
For example, the ``2D Ising model with fields'' is the family of Hamiltonians of the form
\begin{equation}\label{eq:ising}
 H_G(\sigma) = \sum_{\langle i,j \rangle} J_{ij}\sigma_i \sigma_{j} + \sum_i r_i \sigma_i,
\end{equation}
where $\sigma = \sigma_1,\sigma_2,\dots,\sigma_n$ is a configuration of Ising (i.e.\ two-level) spins $\sigma_i\in\{-1,1\}$ on a 2D square lattice, $\langle i,j\rangle$ denotes neigbouring spins, and $J_{ij}$ and $r_i$ are real numbers specifying the coupling strengths and local fields, respectively.

Here we show that there exist certain spin models, which we call \emph{universal}, whose low energy sector can reproduce the complete physics of \emph{any} other classical spin model.
What does it mean to ``reproduce the complete physics''?
Informally, we say that a spin model with Hamiltonian $H$ \emph{simulates} another one $H'$ if (i)~the energy levels of $H$ below a threshold $\Delta$ reproduce the energy levels of $H'$, (ii)~there is a fixed subset $P$ of the spins of $H$ -- which we call the ``physical spins'' -- whose configuration for each energy level below $\Delta$ reproduces the spin configuration of the corresponding energy level of $H'$, and (iii)~the partition function of $H$ reproduces that of $H'$.
Note that from the partition function one can derive all equilibrium thermodynamical properties of the system.
A \emph{universal} model is then a model that can simulate \emph{any} other spin model.

\section*{Precise Statement of Results}
More precisely, we will denote spin degrees of freedom---discrete or continuous---by a string of spin states $\sigma = \sigma_1,\sigma_2,\dots,\sigma_n$.
For $q$-level Ising spins (i.e.\ discrete degrees of freedom with a finite number $q$ of distinct states), we can label the states arbitrarily by integers: $\sigma_i \in \{1,\dots,q\}$.
For continuous spins, a spin state is represented by a unit vector: $\sigma_i \in S^D$ (where $S^D$ is the $D$-dimensional unit sphere; often $D=2$ or~3).
We write $\sigma_R$ to refer to the configuration of a subset $R$ of the spins.
We will often refer to the Hamiltonian being simulated as the \emph{target} Hamiltonian.
With this, we can define simulation more precisely:

\noindent
Let $\sigma'=\sigma'_1,\sigma'_2,\dots$ be the spin degrees of freedom of a target Hamiltonian $H'$.
We say that a spin model with spin degrees of freedom $\sigma=\sigma_1,\sigma_2,\dots$ can \emph{simulate} $H'$ if it satisfies \emph{all three} of the following:
\begin{enumerate}
\item For any $\Delta > \max_{\sigma'} H'(\sigma')$ and any $0 < \delta < 1$, there exists a Hamiltonian $H$ in the model whose low-lying energy levels $E_\sigma=H(\sigma)<\Delta$ approximate the energy levels $E'_{\sigma'}=H'(\sigma')$ of $H'$ to within additive error $\delta$.
 \label{part:eigenvalues}
\item For every spin $\sigma'_i$ in $H'$, there exists a fixed subset $P_i$ of the spins of $H$ (independent of $\Delta$) such that states of $\sigma'_i$ are uniquely identified with configurations of $\sigma_{P_i}$, such that $|E'_{\sigma'} - E_\sigma| \leq \delta$ for any energy level $E_\sigma<\Delta$. We refer to the spins $P=\cup P_i$ in the simulation that correspond to the spins of the target model as the ``physical spins''.
 \label{part:eigenstates}
\item The partition function $Z_H(\beta) = \sum_\sigma e^{-\beta H(\sigma)}$ of $H$ reproduces the partition function $Z_{H'}(\beta) = \sum_{s'} e^{-\beta H'(s')}$ of $H'$ up to constant rescaling, to within arbitrarily small error: $Z_{H'}(\beta) = \gamma (1+\delta) Z_{H}(\beta) + O(e^{-\Delta})$ for some known constant $\gamma$.
 \label{part:partition_function}
\end{enumerate}

Rescaling of the partition function must necessarily be permitted in the definition of simulation, as the universal model will in general have more degrees of freedom than the target model.
(Note that the magnitude of the rescaling has no impact on the efficiency of the simulation, as long as the rescaling is a known, constant value.)
The following trivial example makes this clear.
Consider adding a single $q$-level spin to a system, which does not interact with anything else.
Clearly this new system simulates the original one (just ignore the extra particle).
However, its partition function is rescaled by a factor of $q$.

We say that a model is \emph{universal} if, for any Hamiltonian $H' = \sum_{I=1}^m h_I$ on $n$ spins composed of $m$ separate $k$-body terms, $H'$ can be simulated by some Hamiltonian $H$ from the model specified by $\poly(m,2^k,1/\delta)$ parameters and acting on $\poly(n,m,2^k,1/\delta)$ spins, with the size of the set of physical spins scaling at most as $|P| = \poly(1/\delta)$.

In the following, we establish necessary and sufficient conditions for a model to be universal.
We use these to show that one of the simplest and most widely studied spin models, the 2D Ising model with fields, is universal.
Note that our definition of simulation is very strong: it requires that the target model can be approximated with an arbitrarily large energy cut-off $\Delta$ and to arbitrarily good accuracy $\delta$.
In general, this accuracy is achieved at the expense of increasing the coupling strengths in the universal model.
When both the universal model and the target Hamiltonian have discrete degrees of freedom (e.g.\ the 2D Ising model with fields), the energy levels and configurations are reproduced \emph{exactly}, i.e.\ $\delta=0$.

The first property required for a universal spin model concerns the computational complexity of the ground state energy problem (\textsc{GSE}) of the model.
In this problem, we are asked whether the ground state energy of the system is below some given value $K$.
Recall that NP is the class of Yes/No problems for which every ``Yes'' instance has a ``certificate'' or ``proof'' that can be verified in polynomial time, whereas all certificates are rejected in polynomial time if it is a ``No'' instance.
A problem in NP is NP-complete if every other problem in NP can be efficiently transformed into it (a `polynomial-time reduction').
A canonical NP-complete problem is \textsc{SAT}, which asks whether there is an assignment to the variables of a Boolean formula for which the formula evaluates to \texttt{true} (is `satisfiable').

It is a classic result that the \textsc{GSE} for general spin models is NP-complete~\cite{Bar82}.
This can be seen by providing a polynomial-time reduction from \textsc{SAT}.
That is, given a Boolean formula $\phi$, one constructs a Hamiltonian $H$ such that $\phi$ is satisfiable if and only if there is a spin configuration $\sigma$ such that $H(\sigma)\leq K$, where $H$ and $K$ are determined by $\phi$, and the number of spins and parameters in $H$ is at most polynomially larger than the number of Boolean variables.
For universality we will need a slightly stronger form of reduction: a \emph{faithful} reduction from \textsc{SAT}.
This should additionally preserve the structure of the solution, in the sense that every satisfying assignment of $\phi$ should be in one-to-one correspondence with a ground state configuration of $H$, when the latter is restricted to a subset of spins $P$.

The second condition, which we call \emph{closure}, concerns combining different Hamiltonians from the same model.
We say that a model is \emph{closed} if, for any pair of Hamiltonians $H^{(1)}_{A}$ and $H^{(2)}_{B}$ in the model acting on arbitrary sets of spins $A$, $B$ (which could overlap) there is another Hamiltonian $H$ in the model that simulates $H^{(1)}_{A} + H^{(2)}_{B}$.
If the model places no constraints on the interaction pattern or coupling strengths (for example, it is the set of Ising models on any graph), then it is trivially closed.
Closure is non-trivial if the interaction pattern of the spins is restricted in some way (e.g.\ to a lattice).

Our first main result is that: \textit{A spin model is universal if and only if it is closed and its ground state energy problem admits a faithful, polynomial-time reduction from \textsc{SAT}.}
Our second main result is that the 2D Ising model with fields (see eq.~\eqref{eq:ising}) satisfies these two conditions, hence: \textit{The 2D Ising model with fields is universal.}

\section*{Universality construction}
The intuition behind our results is that closure allows large systems to be built up by combining basic building blocks, whereas the faithful \textsc{SAT} reduction guarantees that the model is sufficiently rich.
More precisely, the \textsc{SAT} reduction allows us to encode universal computation into the ground state.
We use this to isolate one bit of information about the spin configuration, and localise it in a single ``flag spin''.
To give energy $E$ to a particular spin configuration $\sigma$, we make the state of the flag spin indicate whether or not the other spins are in the state $\sigma$ (say, spin-up if they are, spin-down if not).
Adding a term to the Hamiltonian that gives energy $E$ to the spin-up state of the flag spin produces the desired energy level for the $\sigma$ configuration.
We do this for every energy level of each local interaction of the target Hamiltonian, and then combine all the resulting Hamiltonian terms using closure.

To see how this works, consider the example of the Ising model with fields on an arbitrary graph.
First, recall that any boolean expression can be rewritten in terms of a conjunction (boolean AND, $\land$) of clauses, where each clause is the disjunction (boolean OR, $\lor$) of three variables or their negations (boolean NOT, $\neg$).
E.g.\ consider the boolean function $\phi_{00}(x_1,x_2,x_3) = 1$ if $(x_1,x_2,x_3) = (0,0,1)$ or $(0,1,0)$ or $(1,0,0)$ or $(1,1,0)$, and $\phi_{00}(x_1,x_2,x_3) = 0$ otherwise.
We can write this as
\begin{equation}
  \phi_{00}(x_1,x_2,x_3)
  = ( x_1 \lor x_2 \lor x_3) \land
    ( x_1 \lor \neg x_2 \lor \neg x_3) \land
    (\neg x_1 \lor  x_2 \lor \neg x_3) \land
    (\neg x_1 \lor \neg x_2 \lor \neg x_3).
 \label{eq:phi-main}
\end{equation}

Combining results from \cite{Bar82,Ga79} we obtain a faithful reduction from any boolean formula to the Ising model with fields (see also \cite{Gu12}).
First, we identify each boolean variable $x_i$ (taking values $0/1$) with an Ising spin $\sigma_i = 2x_i-1$ (taking values $-1/+1$).
We introduce an additional spin $\sigma_{\neg i}$ (whose state will be identified with $\neg x_i$) coupled to $\sigma_{i}$  by an antiferromagnetic Ising interaction:
\begin{equation}
 h_i = \frac{1}{2} \sigma_i\sigma_{\neg i} + \frac{1}{2}.
\end{equation}
For each clause $c$, we introduce three auxiliary spins $\sigma^{(c)}_1,\sigma^{(c)}_2,\sigma^{(c)}_3$, coupled to each other in a triangle of Ising interactions and local fields, and couple $\sigma^{(c)}_i$ to the spin corresponding to the literal (i.e.\ variable or its negation) appearing in the $i$th position of the clause.
  The Hamiltonian for the clause $(x_1\lor x_2 \lor x_3)$ is
\begin{equation}
  h_c = -\sum_{i=1}^3 \sigma^{(c)}_i - \frac{1}{2}\sum_{i=1}^3\sigma_i
        + \frac{1}{2} \sum_{\begin{subarray}{c} i,j=1\\ i<j \end{subarray}}^3
          \sigma^{(c)}_i \: \sigma^{(c)}_j
        + \frac{1}{2} \sum_{i=1}^3 \sigma_{i}\: \sigma^{(c)}_i + \frac{5}{2}.
\end{equation}
For a clause involving a negated variable $\neg x_i$, the next-to-last term would contain $\sigma_{\neg i}$ instead of $\sigma_i$.

For example, the Hamiltonian for the formula $\phi_{00}$ from eq.~\eqref{eq:phi-main} is
\begin{equation}
 H_{00} = \sum_{i=1}^{3} h_{i} + \sum_{c=1}^{4} h_c ,
 \label{eq:H}
\end{equation}
where the second sum runs over the four clauses in eq.~\eqref{eq:phi-main} (see Fig.~\ref{fig:Fig3}).
One can easily verify that the four ground states of $H_{00}$ have energy~0, and in any ground state, $\sigma_{3}=1$ only when $(\sigma_{1},\sigma_{2})=(-1,-1)$, otherwise $\sigma_{3}=-1$.
Thus, $\sigma_{3}$ acts as a ``flag'' for configuration $(-1,-1)$ of the first two spins.

A flag spin $b$ for a general spin configuration $(\sigma_1,\sigma_2,\dots,\sigma_n)$ is constructed in the same way from the boolean function
\begin{equation}
 \phi(x_1,x_2,\dots,x_n,b) =
 \begin{cases}
  1 & b=1 \text{ and } x_i = \frac{1+\sigma_i}{2} \text{ for all } i\\
  0 & \text{otherwise}.
 \end{cases}
\end{equation}
We ``penalise'' configurations where the flag spin is incorrect by multiplying the whole Hamiltonian by a large constant $\Delta$.

Finally, to reproduce one energy level of the target model, we add a local field $E'$ to the corresponding flag spin, where $E'$ is chosen to be the energy of the target model in this configuration.
For example, to produce an energy level $E'$ for the spin configuration $(-1,-1)$, we use
\begin{equation}
 H= \Delta H_{00} + E'\:\frac{\sigma_{3}+1}{2}.
 \label{eq:H4}
\end{equation}
One can verify directly that this has exactly one energy level below $\Delta$ in which $(\sigma_1,\sigma_2) = (-1,-1)$, and this has energy $E'$.

This construction already shows that the Ising model with fields on an arbitrary graph (with real coupling strengths and fields) is universal.
This is easier than the general case, because closure is trivial (if $H_1$ and $H_2$ are both Ising Hamiltonians then $H_1 + H_2$ is also an Ising Hamiltonian); the model includes arbitrary local field terms, which lets us simply add the second term in \eqref{eq:H4}; and the construction automatically reproduces the partition function (with $\gamma=1/576^{m}$ for a target hamiltonian with $m$ two-body terms).
Nonetheless, the general proof~ \cite{SM} uses the same ideas, is also constructive, and similarly only introduces a polynomial overhead in the number of spins and interactions.
The above proof can be extended to the 2D Ising with fields (with inhomogeneous couplings and local fields), thereby showing that the 2D Ising model with fields is universal.
Full technical details, including precise expressions for the polynomial simulation overhead, are given in the Supplementary Material.

The fact that the Ising model with fields is closed and has a faithful reduction from \textsc{SAT} shows that these conditions are necessary for universality: if it is universal, one of the models it must be able to simulate is the Ising model.

We have focused so far on models with discrete spin degrees of freedom.
For universal models with continuous spins, the requirement of a faithful reduction from \textsc{SAT} implies that the model is able to approximate energy distributions that are sharply peaked around a set of configurations that correspond to discrete boolean values.
Thus the arguments for the discrete case essentially go through unchanged.
(Such sharply peaked energy distributions seem somewhat artificial, so our results suggest that natural models with continuous degrees of freedom will typically not be universal.)
To simulate target models with continuous spins, the idea is to discretize the continuous degrees of freedom sufficiently finely to give a good approximation, and then simulate this discretized version.
One can show that the overhead from this discretization scales favourably with the precision of the approximation \cite{SM}.

\section*{Discussion}
The role of universal models for classical spin Hamiltonians is analogous to that of universal Turing Machines for classical computation (Fig.~\ref{fig:utm-uhs}).
Just as choosing the input to a universal Turing Machine allows it to simulate any other computation, choosing the parameters of a universal model allows it to simulate any other Hamiltonian $H'$.
Moreover, our proof is constructive: it provides the parameters needed for the universal model to simulate $H'$---indeed, these parameters can be computed efficiently from the description of $H'$.

The existence of universal models has intriguing implications for our theoretical understanding of classical many-body physics.
It means that, in different parts of its phase diagram, a universal model will reproduce every phase of every other spin model.
In some sense, this can be viewed as an inversion of the usual renormalisation group flow: by ``fine-graining'' the model and introducing additional short-range parameters, all models (including models in different universality classes) are revealed to be specific cases of the universal model (see Fig.~\ref{fig:idea-RG}).

One might have assumed that the physics of a many-body system would have characteristics determined by the number of spatial dimensions, or the structure of the interaction pattern, or whether the interactions are 2-body or many-body, or the symmetry of its interactions, or whether the local degrees of freedom are continuous or discrete.
The existence of universal models implies that, for models with inhomogeneous couplings, none of these are the case.
E.g.~there can be nothing uniquely characteristic to the physics of spin systems in three spatial dimensions, since all physical properties of a 3D system can be simulated in a 2D Ising model with inhomogeneous couplings. Another way of viewing this is that the inhomogeneity of the couplings destroys all properties related to spatial dimension, symmetries, etc.
Similarly, in a sense there can be no physical properties specific to models with continuous degrees of freedom, since these can be simulated to any desired precision using Ising spins.
Our result also implies and explains the recently found ``completeness results'' \cite{Va08,Ka12b,De09a,De09b,Ka12,Xu11}, where a model is called ``complete'' if its partition function can equal (up to a factor) the partition function of any other model~\cite{Va08}.
Any universal model is complete by taking $\Delta$ to $\infty$ \cite{SM}.

Simulating a target model using a universal model necessarily incurs some overhead.
In general, this is unlikely to be useful for numerical computations, since simulating the target model directly will usually be more efficient.
However, sometimes one is interested in \emph{physically constructing} a particular spin model.
Indeed, in the quantum setting this is precisely the goal of ``analogue'' or ``physical'' Hamiltonian simulation which has already been demonstrated experimentally in specific cases~\cite{Na12b}.
Universal models are potentially interesting in this setting of physical simulation.
Our results imply that one is free to choose a universal model with interactions that are easier to generate in a physical system (e.g.\ two-body, nearest-neighbour interactions), and this will be sufficient to simulate \emph{any} other spin model, even ones with interactions that are difficult to produce directly (e.g.\ long-range, many-body interactions).
Examples of simulation of specific models with their precise polynomial overheads are given in the Supplementary Material.
These constructions can be optimised for a specific pair of universal and target model. Finding optimal constructions in different physical settings is an interesting open problem.

It is not easy in general to determine whether the \textsc{GSE} of a model has a reduction from \textsc{SAT}, but for models with 2-level spins a complete classification is known \cite{Cr92,Sc78}.
This could lead to a classification of all universal models with 2-level spins.
Another interesting possibility is whether translationally-invariant models can be universal.
The NP-hardness of the \textsc{GSE} of certain translationally-invariant models suggests this may be a possibility \cite{Go10b}.

Finally, it is interesting to ask whether such (efficient) universal models exist for the simulation of quantum Hamiltonians, particularly given the potential applications to physical Hamiltonian simulation.
The QMA-completeness of various important quantum Hamiltonians~\cite{Sc09e,Ch13c} and the ability to efficiently simulate Hamiltonian dynamics~\cite{Ll96}, could pave the way to quantum generalisations of our results.

 \section*{Acknowledgements}
 We thank D.
 P\'erez-Garc\'ia, W.\ D\"ur and M.\ Van den Nest for discussions.
 GDLC acknowledges support from the Alexander von Humboldt foundation and SIQS.
 TSC is supported by the Royal Society.

 \noindent
 This work was made possible through the support of grant \#48322 from the John Templeton Foundation.
 The opinions expressed in this publication are those of the authors and do not necessarily reflect the views of the John Templeton Foundation.

 The authors declare no competing financial interests.

\newpage

 \begin{figure}[t]
  \centering \includegraphics[width=0.7\textwidth]{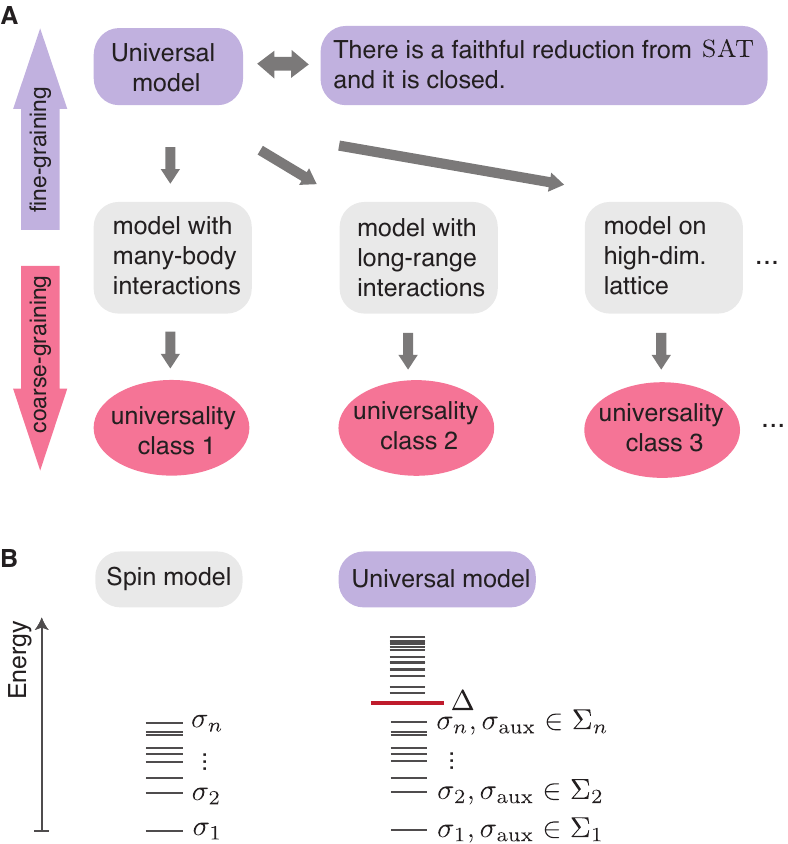}
  \caption{({\bf A}) A spin model is universal if and only if it its ground state energy problem admits a faithful reduction from \textsc{SAT} and it is closed.
    This means that there exists a fine-graining procedure from any spin model ---including models with many-body interactions, long-range interactions, or defined on high dimensional lattices--- that transforms it to the low-energy sector of the universal model.
    On the other hand, coarse-graining different spin models leads to a classification into different universality classes.
    ({\bf B}) For any spin Hamiltonian, the parameters of the universal model can be chosen so that its spectrum below a threshold $\Delta$ is identical to the entire spectrum of the target Hamiltonian.
    Moreover, the spin configuration of each energy level, $\sigma_j$, is reproduced in a subset of the spins of the universal model.
    Finally (not shown) the two partition functions are identical up to a rescaling and an exponentially small factor in $\Delta$.
  }
  \label{fig:idea-RG}
 \end{figure}

 \begin{figure}[t]
  \centering \includegraphics[width=0.75\textwidth]{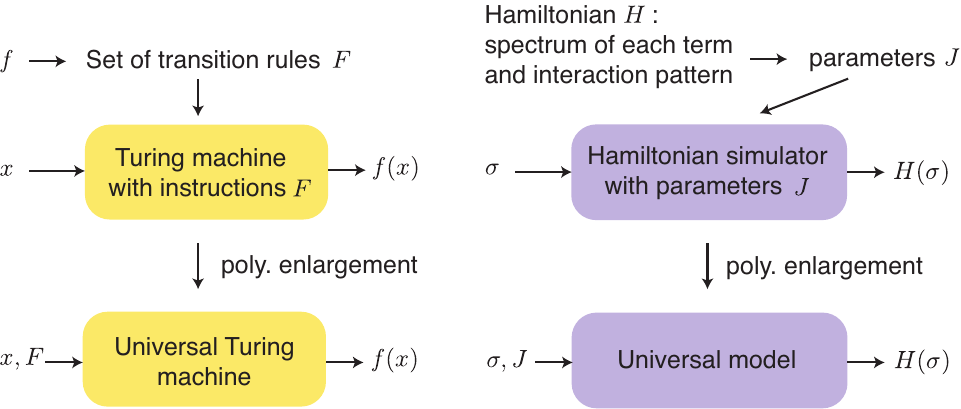}
  \caption{Computation vs.
    Hamiltonian simulation.
    (Left) Any given Turing machine computes some function $f(x)$ of its input $x$.
    Knowing the transition rules of a specific Turing Machine, we can choose the input to a Universal Turing machine in such a way that it simulates the original machine (at the expense of a polynomial overhead).
    (Right) Similarly, any given spin Hamiltonian assigns an energy $H(\sigma)$ to each spin configuration $\sigma$.
    Knowing the energy levels of the individual terms of a specific Hamiltonian, we can choose the parameters of a universal model in such a way that it simulates the original Hamiltonian (at the expense of a polynomial enlargement).}
  \label{fig:utm-uhs}
 \end{figure}

 \begin{figure}[t]
  \centering \includegraphics[width=0.80\textwidth]{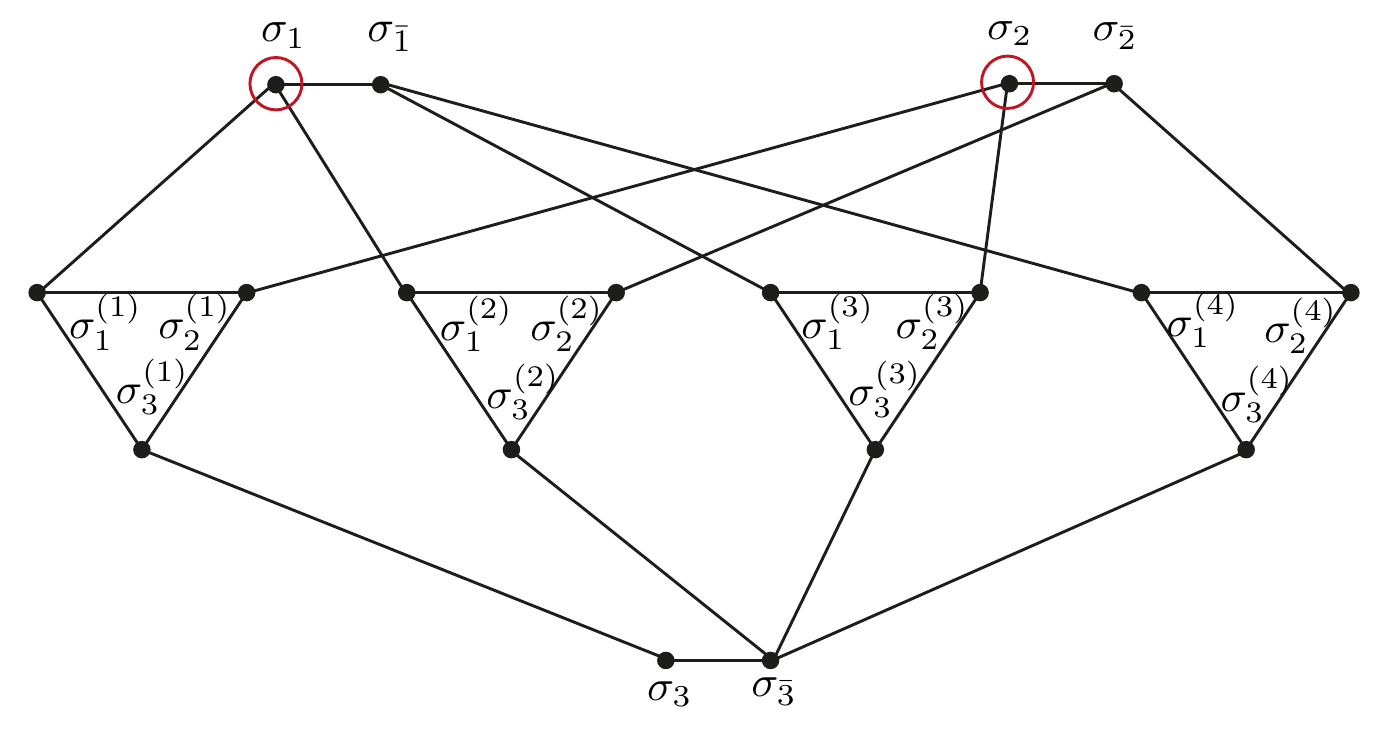}
  \caption{The Ising model with Hamiltonian $H_{00}$ (eq.~\eqref{eq:H}) corresponding to formula $\phi_{00}$ (eq.~\eqref{eq:phi-main}) is defined on the graph shown here. The spins marked in red are those belonging to the physical set $P$.}
  \label{fig:Fig3}
 \end{figure}


\clearpage
\newpage

%
%
%
%
%
%
%
%
%
%



\setcounter{figure}{0}
\makeatletter
\renewcommand{\thefigure}{S\@arabic\c@figure}
\makeatother

\setcounter{equation}{0}
\makeatletter
\renewcommand{\theequation}{S\@arabic\c@equation}
\makeatother

\begin{center}

{\bf\LARGE Simple universal models capture all classical spin physics} \\
\vspace{.5cm}
{\Large Supplementary Material }

\vspace{1cm}

{\large Gemma De las Cuevas$^{\ast 1}$ and Toby S. Cubitt$^{2}$}

\vspace{.5cm}

{\small$^1$Max Planck Institute for Quantum Optics, Hans-Kopfermann-Str.\ 1, D-85748 Garching, Germany\\
$^\ast$To whom correspondence should be addressed. E-mail: gemma.delascuevas@mpq.mpg.de\\
$^2$Department of Computer Science, University College London, 
  Gower Street, London WC1E 6EA, UK}

\vspace{1cm}

\end{center}

In this Supplementary Material, we give additional technical details and mathematical proofs of the results reported in the main text, as well as a number of examples illustrating the general results by applying them to specific spin models.
The material is structured as follows.
\Cref{sec:prelims} introduces the necessary notation and reviews some basic concepts used later.
\Cref{sec:simulation} discusses and gives precise mathematical definitions of Hamiltonian simulation and universal models, as well as introducing the two properties of a spin model (closure and faithful reduction) that will turn out to characterise all universal models.
\Cref{sec:Ising} gives a full that one of the most important spin models, the 2D Ising models with fields, has both of these properties.
\Cref{sec:discrete} gives full technical details of our main result in the case of discrete spins: the closure and faithful reduction properties characterise all universal spin models.
\Cref{sec:continuous} extends this result to the case of continuous spin degrees of freedom.
Finally, \Cref{sec:examples} applies these results to a number examples, showing how various important spin models can be simulated by specific universal models. The examples also illustrate how the general construction can often be simplified when applied to specific universal or target models.

\section{Preliminaries}\label{sec:prelims}
We start by fixing some notation and definitions.
We will need some basic notions from graph theory.
A \emph{graph} is a pair of sets $G=(V,E)$, such that the elements of $E$ are 2-element subsets of $V$ \cite{Di01b}.
$V$ is called the \emph{vertex} set and $E$ the \emph{edge} set.
$G$ is called a \emph{hypergraph} if the elements of $E$ (called hyperedges) contain an arbitrary number of elements of $V$.
Given two graphs $G=(V,E)$, $G'=(V',E')$, their \emph{union} is defined as $G\cup G':=(V\cup V', E\cup E')$, and their \emph{intersection} as $G\cap G'=(V\cap V', E\cap E')$ \cite{Di01b}.
A graph $G$ is said to be a \emph{minor} of a graph $G'$ if $G$ can be obtained as a sequence of edge contractions (shrinking edges and merging their endpoints) and vertex and edge deletions (removing vertices and edges) applied to $G'$ \cite{Bo98}.

We will call a family of spin Hamiltonians a ``spin model'' (or simply ``model'').
Typically, these families will be infinite, and Hamiltonians in the same family (``model'') will be related in some way.
For example, the planar Ising model is the set of all Hamiltonians with 2-body Ising-type interactions between pairs of spins on a planar interaction graph.
The 2D Ising model additionally restricts the interaction graph to be a two-dimensional square lattice.
Thus, a Hamiltonian in the family of the 2D Ising model with fields is specified by the size of the lattice and the set of coupling strengths between neighbouring Ising spins in the lattice.
Since these particular models will recur frequently, we define them formally here:

\begin{definition}[Ising model with fields]
  \label{def:Isingmodel}
  The ``Ising model with fields'' is the family of Hamiltonians specified by a graph $G=(V,E)$ (with vertex set $V$ and edge set $E$), a \emph{coupling strength} $J_{i,j}\in\R$ for each edge $(i,j)\in E$, a \emph{local magentic field strength} $r_i\in\R$ for each vertex $i\in V$, and a \emph{global energy shift} $K$. An Ising spin $\sigma_i\in\{-1,1\}$ is associated with each vertex $i\in V$, and the Hamiltonian is given by:
  \begin{equation}
    H (\{\sigma_i\}) = \sum_{(i,j)\in E} J_{i,j} \: \sigma_i\sigma_j + \sum_{i \in V} r_i \:\sigma_i + K.
    \label{eq:2DIsing}
  \end{equation}
\end{definition}
(On a cubic lattice with constant fields, this model is sometimes called the Edwards--Anderson model \cite{Me09}.)

\begin{definition}[2D Ising model with fields]
  \label{def:2DIsing}
  The ``2D Ising model with fields'' is the subfamily of the Ising model with fields (\cref{def:Isingmodel}) in which $G$ is restricted to be a 2D square lattice.
\end{definition}
Note that the size of the graph is not fixed in either definition.
(We will occasionally refer to this for brevity as the ``2D Ising model'' where it is clear that the model we are referring to includes the local field terms.)

The \emph{interaction pattern} of a model specifies which spins interact with which.
If spins are associated to vertices $V$, every interaction pattern can be represented by a hypergraph.
A hyperedge containing $k$ vertices would represent a $k$-body interaction, i.e.\ a term in the Hamiltonian that depends on only $k$ of the spins \cite{De13b}. 
For example, in \cref{eq:2DIsing}, each term $J_{i,j}\: \sigma_i\sigma_j$ is a 2-body interaction, whereas the term $r_i \: \sigma_i $ is a single-body interaction.
Throughout this text, we will mostly give examples for models with 1- and 2-body interactions, but the results hold generally.

Given a set of spins $\{\sigma_i\in\{1,2,\ldots, q\}\}$ for $i\in\{1,\ldots,n\}$, a \emph{spin configuration} is an assignment of values to the spin variables, e.g.\ $\sigma_1=1$, $\sigma_2=1$, etc.
We denote the set of variables $\{\sigma_i\}$ simply by $\sigma$ for brevity.
We denote by $H(\sigma)$ the energy of configuration $\sigma$ under Hamiltonian $H$.
$H_R$ denotes a Hamiltonian acting on a set of spins $R$.
If $\sigma$ is a configuration of all the spins, and $R$ is a subset of the spins, then $\sigma_R$ denotes the restriction of $\sigma$ to the subset $R$.

We will also need some basic notions of Boolean logic, which we collect here.
A \emph{Boolean variable} $x$ is a variable that can take two truth values, \texttt{true} or \texttt{false}, which we often identify with 1 and 0, respectively.
Boolean variables can be connected using Boolean connectives such as $\lor$ (logical or), $\land$ (logical and).
$\neg x$ will denote the logical negation of variable $x$ (i.e.\ the operation that inverts 0 and 1).
A \emph{Boolean expression} or \emph{formula} is a combination of Boolean connectives and Boolean variables.
$x_i$ or $\neg x_i$ is called a literal.
Given two Boolean expressions $\phi_1$, $\phi_2$, $\phi_1\lor \phi_2$ is called the \emph{disjunction} of $\phi_1$ and $\phi_2$, and $\phi_1\land \phi_2$ the \emph{conjunction} of these two.

A \emph{truth assignment} (or simply assignment) is an assignment of truth values $\{$\texttt{true}, \texttt{false}$\}^n$ to a finite set of Boolean variables $x_1,\dots,x_n$.
An assignment is \emph{satisfying} for an expression $\phi$ if $\phi$ evaluates to \texttt{true} for that assignment.
A Boolean expression $\phi$ is \emph{satisfiable} if there exists a truth assignment to its variables such that $\phi$ evaluates to true.

A Boolean expression $\phi$ is in \emph{conjunctive normal form} (CNF) if $\phi = \bigwedge_{i=1}^m C_i$ and each of the $C_i$s is the disjunction of one or more literals.
Every such $C_i$ is called a \emph{clause}.
If every clause has three literals, then $\phi$ is said to be in \emph{3 conjunctive normal form} (3CNF).
Every Boolean formula can be expressed in 3CNF \cite{Pa95}.

An $n$-ary \emph{Boolean function} is a function $g:\{\texttt{true},\texttt{false}\}^n \to \{\texttt{true},\texttt{false}\}$.
A Boolean expression $\phi$ with variables $x_1,\ldots, x_n$ expresses the $n$-ary Boolean function $g$ if, for any $n$-tuple of truth values $t=(t_1,\ldots, t_n)$, $g(t)$ is \texttt{true} if $t$ satisfies $\phi$, and $g(t)$ is \texttt{false} if $t$ does not satisfy $\phi$, where $T(x_i)=t_i$ for $i=1,\ldots, n$.
Every Boolean expression expresses some Boolean function.
Conversely, any $n-$ary Boolean function $g$ can be expressed as a Boolean expression $\phi_g$ involving variables $x_1,\ldots, x_n$.

We will denote by $\{e_\ell(x):\mathbb{Z}_{2}^{k}\to \mathbb{Z}_2\}$ the elementary basis for Boolean functions on $k$ bits, $e_\ell(x) = 1$ if and only if (iff)  $x=\ell$, and $e_\ell(x) = 0$ otherwise.
Finally, $x_1\odot x_2$ will denote the Boolean exclusive-nor operation, i.e.\ $x_1\odot x_2=1$ iff $x_1=x_2$.

\section{Hamiltonian Simulation}\label{sec:simulation}
We are now in a position to give a precise definition of Hamiltonian simulation.
When showing that one spin model can simulate another, we call the model we are simulating the ``target model'', and a specific Hamiltonian in that model the ``target Hamiltonian''.
We refer to a model that can simulate \emph{any} other spin model as a ``universal model''.

We will first state and prove our main result under the assumption that the spins of the \emph{target} Hamiltonians are discrete $q$-level Ising spins (i.e.\ they take values in $\{1,\dots,q\}$).
We then generalise the result to continuous classical spin models, in which the spins can rotate continuously, i.e.\ the spins are normalised vectors in $\mathbb{R}^D$, which can equivalently be thought of as points on the $D$-dimensional unit sphere $S^D$.
(In fact, our proof applies equally well to models whose local degrees of freedom take values in an arbitrary compact set, but we will restrict to spins for clarity.)

We will also assume that the spins of the universal Hamiltonian are 2-level Ising spins.
However, unlike the case of continuous \emph{target} models, which requires additional work to prove, the assumption of discrete \emph{universal} models is purely for notational simplicity; the argument for continuous universal models is identical.

\begin{definition}[Hamiltonian Simulation -- discrete case]
  \label{def:simulation}
  We say that a spin model can \emph{simulate} a Hamiltonian $H'$ on Ising spins if it satisfies \emph{all three} of the following:
  \begin{enumerate}
  \item For any $\Delta > \max_{\sigma'} H'(\sigma')$ there exists a Hamiltonian $H$ from the model whose low-lying energy levels $\{E_\sigma=H(\sigma):E_\sigma<\Delta\}$ are identical to the energy levels $\{E'_{\sigma'}=H'(\sigma')\}$ of $H'$.
    \label[part]{part:eigenvalues}
  \item There is a one-to-one correspondence $F(\sigma_{P_i}) = \sigma'_i$ between the states $\sigma'_i$ of each spin in $H'$ and the spin-configurations $\sigma_{P_i}$ of a fixed subset $P_i$ of spins in $H$, such that for all configurations $\sigma$ with energy $E_\sigma<\Delta$, the energy $E'_{\sigma'} = E_\sigma$ for $\sigma'_i = F(\sigma_{P_i})$.
    \label[part]{part:eigenstates}
  \item The partition function $Z_H(\beta) = \sum_\sigma e^{-\beta H(\sigma)}$ of $H$ reproduces the partition function $Z_{H'}(\beta) = \sum_{\sigma'} e^{-\beta H'(\sigma')}$ of $H'$ up to rescaling, to within exponentially small additive error: $Z_{H'}(\beta) = \gamma Z_{H}(\beta) + O(e^{-\Delta})$ for some known constant $\gamma$. (Here, $\beta$ is inverse-temperature).
    \label[part]{part:partition_function}
  \end{enumerate}
\end{definition}

We will refer to the set of spins $P:=\cup_{i}P_{i}$ as the ``physical spins''.
Note that the cardinality of $P$ is fixed (it depends on number of spins in $H'$), and, in particular, is independent of $\Delta$.
Note also that, given $H'$, $H$ is fixed, hence the number of terms of both $Z_H(\beta)$ and $Z_{H'}(\beta)$ is constant. The $O(e^{-\Delta})$ term in the partition function approximation thus indicates that the approximation error is at most some constant times  $e^{-\Delta}$, where this constant will generally depend on the number of terms.
This approximation error can be made arbitrarily small at the expense of increasing the strengths of some of the couplings to arbitrarily large values. (This is also the case for the earlier completeness results, see \cref{sec:discrete}).

\begin{definition}[Universal Model]
  \label{def:universal}
  We say that a model is \emph{universal} if, for any Hamiltonian $H' = \sum_{I=1}^m h_I$ on $n$ spins composed of $m$ separate $k$-body terms $h_I$, $H'$ can be simulated by some Hamiltonian $H$ from the model specified by $\poly(m,2^k)$ parameters and acting on $\poly(n,m,2^k)$ spins.
\end{definition}

Note that we demand that the overhead incurred by simulating a target Hamiltonian using a universal model scales polynomially in the number of parameters required to describe the target Hamiltonian.
A general target Hamiltonian $H'$ on $n$ spins can involve $n$-body interactions, requiring $2^n$ parameters to specify.
In this case, any Hamiltonian that simulates it will necessarily have to contain at least $2^n$ parameters.
On the other hand, if the target Hamiltonian is made up of $m$ separate $k$-body terms, it can be described by $m 2^k$ parameters.
In this case, our definition of universality requires that the universal model should only contain $\poly(m,2^k)$ parameters.
In other words, the simulation overhead should be polynomial in all cases.

Note also that we allow for the partition function to be rescaled by a known constant $\gamma$, since $H$ has generally more degrees of freedom than $H'$, so that its partition function may scale differently.
For $\Delta=\infty$, this simply amounts to a global energy shift.

Up to now, we have placed no restriction whatsoever on what a spin model can look like; it could consist of an arbitrary collection of completely unrelated Hamiltonians.
However, in any reasonable spin model we expect there to be some relationship between different Hamiltonians within the same model.
The following definition imposes some additional structure on spin models, such that different Hamiltonians in the same model are at least loosely related.

\begin{definition}[Closed Model]\label{def:closure}
  We say that a spin model is \emph{closed} if, for any pair of Hamiltonians $H^{(1)}_{A}$ and $H^{(2)}_{B}$ in the model acting on arbitrary sets of spins $A$ and $B$ respectively, possibly with $A\cap B\neq \emptyset$ there exists a Hamiltonian in the model which simulates \mbox{$H^{(1)}_{A} + H^{(2)}_{B}$}.
\end{definition}

If the model places no constraints on the interaction graph (e.g.\ it consists of arbitrary $k$-body terms, with no restrictions on which sets of $k$ spins can interact directly, and allows arbitrary inhomogeneous coupling strengths), then it is trivially closed: the sum of two such Hamiltonians is itself another Hamiltonian from the model.
Closure becomes a non-trivial property for spin models in which the form of the interaction graph is restricted in some way (e.g.\ to a planar graph, or to a square lattice).

We will make use of the concept of reductions between decision problems from complexity theory. We first recall the definitions of some of the most famous computational problems in complexity theory, which we will make key use of:

\begin{definition}[\textsc{SAT}]
  \label{def:SAT}
  Given a Boolean expression $\phi$, is it satisfiable?
\end{definition}

\begin{definition}[\textsc{3SAT}]
  \label{def:3SAT}
  Given a Boolean expression $\phi$ in 3 conjunctive normal form, is it satisfiable?
\end{definition}
Since every Boolean expression can be expressed in 3CNF, every \textsc{SAT} problem can be rephrased as a \textsc{3SAT} problem.

The basic computational problem corresponding to the task of finding ground state energies of spin models is called the \textsc{Ground State Energy} problem (\textsc{GSE}):
\begin{definition}[\textsc{Ground State Energy}]
  The ground state energy problem of a model $\mathcal{M}=\{H_\alpha\}_{\alpha}$, asks: given $H_{\alpha}\in \mathcal{M}$, is there a configuration $\sigma$ such that $H_{\alpha} (\sigma) \leq 0$?
\end{definition}
A configuration $\sigma$ that minimizes $H_{\alpha}(\sigma)$ is called a ground state configuration.

A \emph{polynomial-time reduction} from \textsc{SAT} to \textsc{GSE} is a map $f$ from boolean formulae to Hamiltonians, which can be computed in time polynomial in the size of the formula, such that a formula $\phi$ is satisfiable if and only if $f(\phi)$ is a yes-instance of \textsc{GSE} (i.e.~there is a configuration $\sigma$ such that $H_{\alpha}(\sigma)\leq 0$).
We will need a form of polynomial-time reduction between \textsc{SAT} and \textsc{GSE} that additionally preserves the structure of the witnesses.
(We could use another NP-complete problem in the definition instead of \textsc{SAT}, or even define faithful reductions directly in terms of witnesses.
But this would just make the identification between witnesses and spins in the definition more complicated, for little gain.)

\begin{definition}[Faithful reduction]
  \label{def:faithful}
  We say that a reduction from \textsc{SAT} to \textsc{GSE} is \emph{faithful} if there exists a one-to-one mapping $F$ between subsets of spins $R_i$ and Boolean variables $x_i$, such that the ground state configuration $\sigma$ corresponds to a satisfying assignment $x_i = F(\sigma_{R_i})$ of the \textsc{SAT} problem.
   We denote the set of spins that correspond to Boolean variables by $R :=\cup_{i}R_{i}$.
\end{definition}

\begin{remark}[On the relation between faithful reductions and FNP]
  The concept of faithful reduction is related to that of a reduction between the function versions of \textsc{SAT} and \textsc{GSE}.
  \textsc{SAT} and \textsc{GSE} are \emph{decision problems}, in which we are asked whether a solution exists.
  In function problems, we are additionally asked to \emph{provide} the solution, if it exists.
  For example, in the function version of \textsc{SAT}, called \textsc{FSAT}, we are given a Boolean expression $\phi$, as in \textsc{SAT}.
  But now, if $\phi$ is satisfiable, we must return a satisfying assignment for $\phi$; otherwise we must return ``no''.
  Similarly, in \textsc{FGSE}, the function version of \textsc{GSE}, we are given a Hamiltonian $H$, and if there exists a spin configuration $\sigma$ so that $H(\sigma)\leq 0$, we must provide $\sigma$; otherwise we must return ``no''.

  Informally, NP is the class of decision problems such that, if the answer is ``yes'', then there is a proof of this fact (the `witness'), of length polynomial in the size of the input, that can be verified in polynomial time; if the answer is ``no'', then the algorithm must reject all purported proofs.
  FNP, the function version of NP, is the class of function problems such that, if the answer is ``yes'', then the algorithm must return a proof of this that can be verified in polynomial time; if the answer is ``no'', then the algorithm must reject all proofs.
  \textsc{FSAT} is known to be FNP-complete \cite{Pa95,Mo11}.

  Now, a reduction between function problems \textsc{A} and \textsc{B} is a polynomial map $g$ such that $x$ is a yes-instance of \textsc{A} if and only if $g(x)$ is a yes-instance of \textsc{B}, and another polynomial map $G$ such that $w$ is the witness of the fact that $x$ is a yes-instance of \textsc{A} if and only if $G(w)$ is a witness of the fact that $g(x)$ is a yes-instance of \textsc{B}.
  That is, there is an additional polynomial map, $G$, that maps the witness of one problem to a witness of the other.

  Our faithful reduction is thus a reduction between \textsc{FSAT} and \textsc{FGSE} in which $G$ is a particularly simply map, as it merely identifies part of the witness of the \textsc{FGSE} problem (namely the ground state configuration $\sigma_R$ of a subset $R$ of the spins) with the witness of the \textsc{FSAT} problem.
\end{remark}

\section{2D Ising Model with Fields}\label{sec:Ising}
The 2D Ising with fields is an important example of a model that is closed and whose \textsc{GSE} admits a faithful reduction from \textsc{SAT}.
To see this, we first need to recall some well-known computational problems.
Given a graph $G=(V,E)$, a \emph{vertex cover} is a a subset $C\subseteq V$ such that, for each edge $(u,v)\in E$, at least one of $u$ and $v$ belongs to $C$ \cite{Di01b}.
The cardinality of $C$, $|C|$, is also called the \emph{size} of $C$.
Trivially, every graph $G=(V,E)$ has a vertex cover of size $|V|$.

\begin{definition}[\textsc{Vertex Cover}]
  \label{def:VertexCover}
  Given a graph $G=(V,E)$ and an integer $K$, does $G$ have a vertex cover of size $K$ or less?
\end{definition}

A graph $G$ is \emph{planar} if it can be drawn in the plane without any two edges crossing.

\begin{definition}[\textsc{Planar Vertex Cover}]
  \label{def:Planar VertexCover}
  Given a planar graph $G$ and an integer $K$, does $G$ have a vertex cover of size $K$ or less?
\end{definition}

An \emph{independent set} in a graph $G=(V,E)$ is a subset $I\subseteq V$ such that, for all $u,v\in I$ the edge $(u,v)$ is \emph{not} in $I$.
It is easy to see that $C\subseteq V$ is a vertex cover of $G=(V,E)$ iff $V-C$ is an independent set of $G$ \cite{Mo11}.
Similarly, $|I|$ is called the size of the independent set $I$ \cite{Ga79}.

\begin{definition}[\textsc{Independent Set}]
  \label{def:IndSet}
  Given a graph $G=(V,E)$ and an integer $J$, does $G$ have an independent set of size $J$ or more?
\end{definition}

We are now ready to show that the 2D Ising model with fields fulfils our two properties: closure and faithful \textsc{SAT} reduction.
Many different ways of reducing \textsc{SAT} to Ising are known, but not all of these are faithful.
(I.e.\ spin configurations in the ground state do not necessarily correspond to values of the boolean variables for which the formula evaluates to \texttt{true}.)
We perform the reduction in two steps: \textsc{SAT} to \textsc{Vertex Cover}, then \textsc{Vertex Cover} to \textsc{GSE} for the 2D Ising model.
The first step uses a standard construction from Garey and Johnson \cite{Ga79}, the second uses the classic construction of Barahona \cite{Bar82}.
What we need to additionally verify is that each step on these constructions is faithful.

\begin{lemma}
  \label{lem:2DIsing}
  The 2D Ising with fields (\cref{def:2DIsing}) is closed and its \textsc{GSE} admits a polynomial-time faithful reduction from \textsc{SAT}.
\end{lemma}
\begin{proof}
  We first show that the \textsc{GSE} of the 2D Ising with fields admits a
  polynomial-time faithful reduction from \textsc{SAT}, via the following sequence of transformations:
  \begin{enumerate}
  \item
    \emph{Express the \textsc{SAT} formula in 3 Conjunctive Normal Form (3CNF)}.
    \label[step]{3CNF}
    Every Boolean formula can be written in 3CNF, so the \textsc{SAT} problem is translated into a \textsc{3SAT} problem.
    Rewriting a formula in 3CNF gives a boolean formula over the same set of variables with the same satisfying assignments, so this step is manifestly faithful.

  \item
    \emph{Reduce \textsc{3SAT} to \textsc{Vertex Cover}}.
    \label[step]{im:3SAT-VC}
    The construction and argument follows \cite{Ga79} (page 55).
    Given a formula $\phi$ in 3CNF, construct a graph $G = (V,E)$ as follows.
    Let $U=\{u_1,\ldots, u_n\}$ denote the set of variables and $C=\{c_1,\ldots,c_m\}$ the set of clauses of $\phi$.
    Denote the three variables or their negations involved in clause $c_j$ by $x_j,y_j,z_j$.

    The graph $G$ has $2n+3m$ vertices $V = V_U\cup V_C$, labelled $V_U = \{u_i,\bar{u}_i\}_{i=1,\dots,n}$ and $V_C =\{a_x[j],a_y[j],a_z[j]\}_{j=1,\dots,m}$.
    That is, there are two vertices $u_i$, $\bar{u}_i$ for each boolean variable ($\bar{u}_i$ will be associated with the negation $\neg u_i$ of the corresponding boolean variable), and three vertices $a_x[j],a_y[j],a_z[j]$ for each clause $c_j$.

    Each pair of vertices $u_i$ and $\bar{u}_i$ is joined by an edge: $E_U = \{(u_i,\bar{u}_i)\}_{i=1,\dots,n}$.
    Each triple of vertices  $a_x[j],a_y[j],a_z[j]$ is joined by a triangle of edges:
    \begin{equation}
    E_\Delta = \left\{(a_x[j],a_y[j]), (a_y[j],a_z[j]),(a_z[j],a_x[j])\right\}_{j=1,\dots,m}.
    \end{equation}
    For each clause $c_j$ involving variables $x_j,y_j,z_j$, every vertex in the triangle corresponding to that clause is joined by an edge to a vertex corresponding to one of the variables: $E_C = \{(a_x[j],x_j), (a_y[j],y_j), (a_z[j],z_j)\}_{j=1,\dots,n}$.
    (If a variable $x_j$ is negated in the clause, the edge joins the triangle to $\bar{x}_j$ rather than $x_j$.)
    The complete edge set of the graph is then $E = E_U\cup E_\Delta\cup E_C$.
    (\Cref{fig:3SAT-INDSET} shows an example of the resulting graph for the boolean formula given in \cref{eq:phi1k2}.)

    Now, $G$ cannot have an independent set of size greater than $n+m$, since the set can include at most one vertex from each pair $u_i,\bar{u}_i$ and at most one vertex from each triple $a_x[j],a_y[j],a_z[j]$.
    Hence the minimum size of a vertex cover is $|V|-(n+m)=n+2m$ (see \cref{def:IndSet}).
    For an arbitrary assignment to the boolean variables, consider the vertex set $F$ consisting of all the vertices labelled by variables that are \texttt{false} in the assignment (i.e.\ the set includes vertex $u_i$ if that variable is \texttt{false}, or $\bar{u}_i$ if $u_i$ is \texttt{true}).
    $F$ is manifestly an independent set, but is not necessarily maximal.
    If clause $c_j$ is satisfied by the assignment, then at least one variable appearing in $c_j$ ($x_j$, say, without loss of generality) is \texttt{true}, hence does not appear in $F$.
    We can therefore extend the independent set $F$ by adding the triangle vertex $a_x[j]$ to the set.
    (Note that we cannot extend the independent set with any further vertices from the same triangle.)
    On the other hand, if clause $c_j$ is not satisfied, then all variables appearing in $c_j$ are \texttt{false}, all corresponding vertices appear in $F$, and none of the triangle vertices can be added to the independent set.

    If the assignment satisfies $\phi$, then all clauses are satisfied and we obtain in this way an independent set of size $n+m$ which matches the upper bound, hence is maximal.
    If $\phi$ is not satisfiable, then every assignment fails to satisfy at least one clause, and the resulting independent set has size $<n+m$.
    It remains to show that there is no other larger independent set.
    An independent set of size $n+m$ necessarily contains one vertex from each triangle, hence contains $n$ vertices from $V_U$.
    Since this set is independent, it must contain exactly one vertex from each $u_i,\bar{u}_i$ pair, hence can be consistently identified with an assignment to the boolean variables.
    But we have already seen that no such independent set exists if $\phi$ is unsatisfiable.

    The above argument also shows that the reduction is faithful with $R = \{u_i\}_{i=1,\dots,n}$, since vertices in $R$ that are in the vertex cover (resp.\ independent set) are in one-to-one correspondence with \texttt{true} (resp.\ \texttt{false}) variables in the satisfying assignment.

  \item
    \emph{Reduction from \textsc{Vertex Cover} to \textsc{Planar Vertex Cover}}.
    \label[step]{im:planarvc}
    This step follows \cite{Ga76}.
    Take an arbitrary projection of $G$ onto the plane, resulting in a non-planar graph with, say, $t$ crossings.\footnote{Note that determining the projection that gives rise to the minimum number of crossings is itself an NP-hard problem \cite{Ga83}, but we do not need to minimize $t$ here.}
    Replace every crossing by a ``crossing gadget'' \cite{Ga76}, introducing 22 vertices and 40 edges per crossing, as shown in \cref{fig:3SAT-INDSET}.
    This defines a new graph $\tilde{G}=(\tilde{V},\tilde{E})$, which is planar.
    Following \cite{Ga76}, we will call the vertices at the four corners of a crossing gadget \emph{outlet vertices}, or simply \emph{outlets}.

    Garey, Johnson and Stockmeyer prove in \cite{Ga76} that $G$ has a vertex cover of size $p$ iff $\tilde{G}$ has a vertex cover of size $p+13t$ \cite{Ga76}.
    More precisely, they prove that any vertex cover of $G$ can be extended to a vertex cover of $\tilde{G}$ by adding 13 vertices from crossing gadget.
    And, conversely, $\tilde{G}$ always has a minimal vertex cover such that restricting $\tilde{G}$ to the vertices of $G$ gives a vertex cover for $G$.
    (Note that there are no minimal vertex covers of $\tilde{G}$ that contain two opposite outlets from the same crossing gadget (see Table~1 in \cite{Ga76}).)

    We claim that any minimal vertex cover of $\tilde{G}$ that includes exactly one vertex from each pair of opposite outlets (see \cref{fig:3SAT-INDSET}) restricts to a minimal vertex cover of $G$.
    To see this, we need to show that if in the restriction to $G$, any edges in $G$ that were replaced by crossing gadgets in $\tilde{G}$ are covered.
    Note that the vertex cover of $\tilde{G}$ must contain all the vertices of $G$ that are attached to outlet vertices \emph{not} in the cover (otherwise the edge joining that vertex to the crossing gadget outlet would not be covered).
    Thus at least one of the vertices on any edge that was replaced by a crossing gadget is included in the cover, as required.

    We have shown that every minimal vertex cover of $G$ extends to a minimal vertex cover of $\tilde{G}$ contianing exactly one vertex from each outlet pair, and conversely any minimal vertex cover of $\tilde{G}$ with this property restricts to a vertex cover of $G$.
    A vertex cover of $\tilde{G}$ which extends a vertex cover of $G$ clearly preserves faithfulness, since the set $R$ only contains vertices from the original graph $G$.
    \cite{Ga76} does not prove that all minimal vertex covers have this property.\footnote{Though it is conceivable that this is true.}
    However, for our purposes, there is a simpler way of ensuring faithfulness in the next step, which is equivalent to simply imposing at this stage that we must choose a minimal vertex cover that contains exactly one vertex from each pair of opposite outlets.

    \begin{figure}[thb]
      \centering \includegraphics[width=0.9\textwidth]{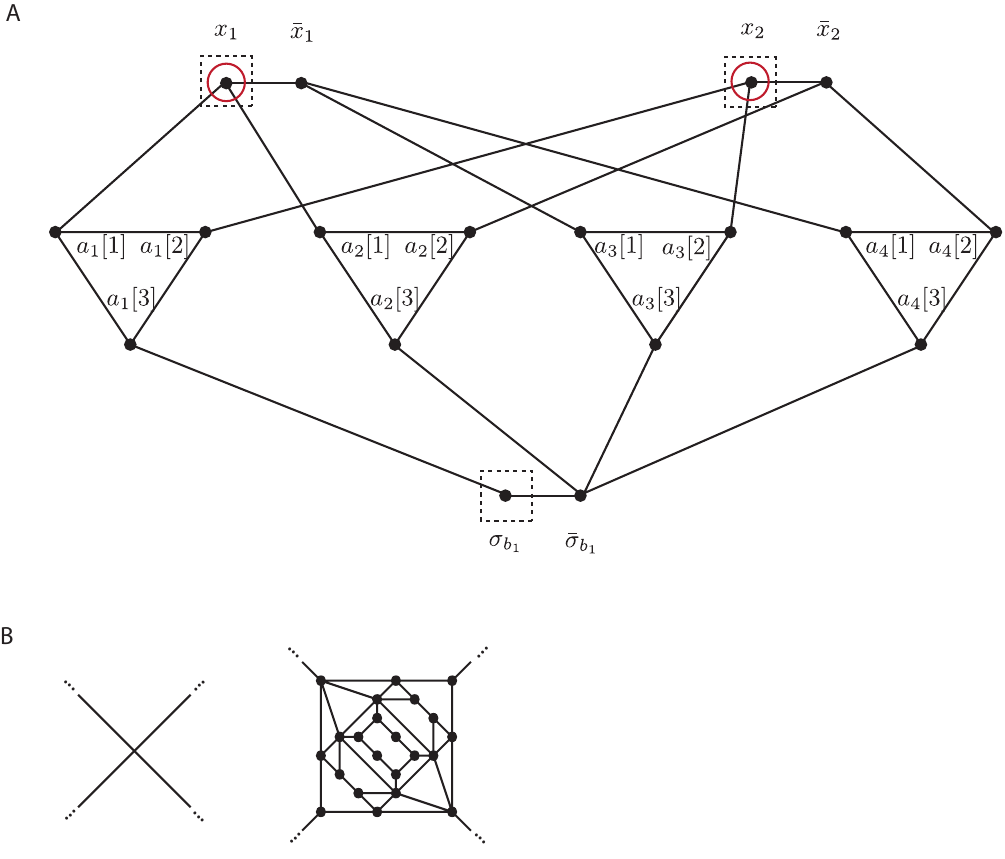}
      \caption{A. The Boolean formula $\phi_1$ in \cref{eq:phi1k2} is satisfiable iff the graph shown here has a vertex cover of size 11 (see the proof of \cref{lem:2DIsing}).
        The variables belonging to $R$ (see \cref{def:faithful}) are marked with dashed squares, and those in the physical set $P$ (see \cref{def:simulation}) are marked with red circles.
        B. The crossing gadget to reduce \textsc{Vertex Cover} to \textsc{Planar Vertex Cover} (\cref{im:planarvc} in the proof of \cref{lem:2DIsing}) \cite{Ga76}.
      }
      \label{fig:3SAT-INDSET}
    \end{figure}

  \item
    \emph{Reduction from \textsc{Planar Vertex Cover} to \textsc{GSE} of Planar Ising with fields}.
    \label[step]{im:PVC-Ising}
    This step in the reduction follows Barahona \cite{Bar82}.
    Given a planar graph $\tilde{G}=(\tilde{V},\tilde{E})$, associate one Ising variable $\sigma_i\in \{-1,1\}$ to each vertex $i\in \tilde{V}$ and define the Ising Hamiltonian
    \begin{equation}
      \label{eq:HGising}
      h_{\tilde{G}}(\{\sigma_i\}) = \frac{1}{2} \sum_{i\in \tilde{V}}\sigma_{i} (1- \textrm{deg}(\sigma_i)) + \frac{1}{2}\sum_{(i,j)\in \tilde{E}} \sigma_i \sigma_j -\frac{|\tilde{V}|}{2} + \frac{|\tilde{E}|}{2} + K,
    \end{equation}
    where $\textrm{deg}(\sigma_i)$ is the degree (i.e.\ number of neighbours) of the vertex associated with $\sigma_i$.
    Since the interaction graph of $h$ is $\tilde{G}$, this is obviously a planar Ising model.
    We claim that $h_{\tilde{G}}$ has a ground state of energy 0 iff $\tilde{G}$ has a vertex cover of size $|\tilde{V}|-K$.

    To see this, define the Boolean variables $s_i\in \{0,1\}$, which are related to the Ising variables by $\sigma_i=1-2s_i$. Then, in terms of this change of variables,
    \begin{equation}
      \label{eq:pseudoising}
      h_{\tilde{G}}(\{s_i\}) = - \sum_{i\in \tilde{V}} s_i + 2 \sum_{(i,j)\in \tilde{E}} s_is_j + K.
    \end{equation}
    It is easy to see that an independent set of size $K$ for $\tilde{G}$ gives a configuration with energy $0$, by setting the $s_i=1$ for vertices that are in the independent set, and setting all other $s_i=0$.
    Moreover, the ground state corresponds in this way to a maximal independent set, since if a set contains two adjacent vertices (thus being an invalid independent set), the energy can be decreased by removing one of them from the set (due to the factor of~2 multiplying the interaction term).

    Note that this reduction is faithful, since vertices in the vertex cover are in one-to-one correspondence with boolean variables $s_i$ taking the value~0, or equivalently, spins $\sigma_i=1$ in the ground state.
    See \cref{tab:witness-ising} for a summary of the correspondences of the witnesses.

    \begin{center}
      \begin{table}[htb]\centering
        \begin{tabular}{c|c|c|c|c}
        \textsc{3SAT} & Vertex Cover $C$ & Ind.~Set $I$ &
        \textsc{GSE} of Ising (\cref{eq:pseudoising})&
        \textsc{GSE} of Ising (\cref{eq:HGising})\\
        \hline
        $x_i=1$           & $i \in C$ & $i \notin I$ &  $s_{i}=0$ & $\sigma_{i}=1$\\
        $x_i=0$           & $i \notin C$ & $i \in I$ &  $s_{i}=1$ & $\sigma_{i}=-1$\\
        \end{tabular}
        \caption{Summary of the identifications of the witnesses in the computational problems
          \textsc{SAT} (where the witness is given by an assignment to the boolean variables $\{x_{i}\}$),
          \textsc{Vertex Cover} (where the witness is the vertex cover $C$ itself),
          \textsc{Independent Set} (where the witness is the independent set $I$ itself; $C$ is a vertex cover     iff $V\backslash I$ is an independent set),
          \textsc{GSE}  of Ising model depending on 0/1 spins (\cref{eq:pseudoising}),
          and \textsc{GSE} of an actual Ising model (\cref{eq:HGising}).
          In summary, a Boolean variable $x_{i}=1$  iff the corresponding Ising variable $\sigma_{i}=1$.
        }
      \end{table}
      \label{tab:witness-ising}
    \end{center}

    It remains to enforce the additional constraint we imposed in \cref{im:planarvc}: that the minimal vertex cover must be one that contains exactly one vertex from each pair of opposite outlet vertices.
    To enforce this, we first multiply the entire Hamiltonian constructed thus far by a large constant factor, $100$ say.\footnote{This is certainly larger than necessary, but it makes the argument clear.}
    We then add local field terms $-\sigma_i$ to all outlet vertices $\tilde{V}_{o}$:
    \begin{align}
      &H_{\tilde{G}} = 100 \: h_{\tilde{G}} + h_{\textrm{outlets}},  \\
      &h_{\textrm{outlets}} = -\sum_{i\in \tilde{V}_{o}} \sigma_{i} .
    \end{align}
    In this way, vertex covers that contain more outlet vertices have lower energy, but the ground state must still be a minimal vertex cover.
    The latter follows from the fact that every additional vertex in the cover increases the energy contribution from the original Hamiltonian by~100.
    Whereas adding an outlet vertex to the cover only reduces the energy contribution of $h_{\textrm{outlets}}$ by~1.
    Since vertex covers that include more than one vertex from an outlet pair are never minimal (see \cref{im:PVC-Ising}), the ground states must correspond to minimal vertex covers containing exactly one vertex from each outlet pair.
    Finally, since all these minimal vertex covers contain the same number of outlet vertices, all such states have the same energy and are ground states.

  \item
    \emph{Reduction from \textsc{GSE} of Planar Ising with fields to \textsc{GSE} of 2D Ising with fields}.
    \label[step]{IsingG->2D}
    Since $\tilde{G}$ is a planar graph, it can be embedded as a minor of a 2D square lattice $F$ with only a polynomial overhead in the number of spins \cite{de88,de90,Sc90}.
    The effects of the edge contraction and deletion operations are easily implemented in the Ising model: contraction of edge $(i,j)$ by setting $J_{ij}> \Delta$, and $h_{i}=0$ (or $h_{j}=0$),
    and deletion of edge $(i,j)$ by setting $J_{ij}=0$ (see \cite{De09a}).
    These operations  preserve faithfulness, as the set of spins $R$ will be left unchanged.
  \end{enumerate}

  It remains to prove that the 2D Ising model with fields is closed.
  Clearly, the Ising model with fields (\cref{def:Isingmodel}) is closed, since for any two graphs $G_1$, $G_2$, their union $G_1\cup G_2=:G$ is another graph $G$ with another set of couplings $\mathcal{J}=\{\mathcal{J}_1,\mathcal{J}_2\}$, which is another element of the family.
  To see that the 2D Ising model with fields is also closed, consider two 2D square lattices $F_1$, $F_2$, possibly with $F_1\cap F_2\neq \emptyset$, which implies that generally $F_1 \cup F_2 =: G'$ is not a 2D square lattice.
  However, there will be another 2D square lattice $F_3$ which will be able to simulate $G'$, as $G'$ can be projected to the plane, each crossing replaced by a crossing gadget as in \cref{im:planarvc}, and the resulting planar graph can be obtained as a minor of $F_3$, as in \cref{IsingG->2D}.
\end{proof}

\section{Universal Hamiltonian Simulation -- discrete spins }\label{sec:discrete}
We are now in a position to prove our main result: that the closure and faithful \textsc{SAT} reduction conditions are necessary and sufficient for a model to be universal. In this section, we state and prove the result for spin models with discrete degrees of freedom ($q$-level Ising spins). \Cref{sec:continuous} extends the result to models with continuous spin degrees of freedom.

\begin{theorem}[Main result -- discrete case]
  \label{thm:universal_characterisation}
  A spin model is universal if and only if it is closed and its \textsc{Ground State Energy} problem admits a polynomial-time faithful reduction from \textsc{SAT}.

  If the faithful reduction from a \textsc{SAT} problem on $n$ spins requires at most $p(n)$ spins, and the universal model is trivially closed, then simulating the energy levels and configurations (\cref{def:simulation}, \cref{part:eigenvalues,part:eigenstates}) of a Hamiltonian on $n$ $q$-level spins with $m$ $k$-body interactions requires at most $m\, q \, 2^{k+1} \, p\left(k\,\lceil\log_2q\rceil+1\right)$ spins and $3\, m\, q\, 2^{k+1}$ interaction terms.
  If closure is non-trivial, the overhead is a polynomial $g\left(m\, q\, 2^{k+1}\, p\left( k\, \lceil\log_2 q \rceil + 1\right) , \, 3\, m\, q\, 2^{k+1}\right)$ in both these quantities, where $g$ depends on the particular universal model.

  Approximating the partition function (\cref{def:simulation}, \cref{part:partition_function}) increases the overhead to at most $m\, q\, 2^{k+1} \left[ \left( q \, 2^{k+1}-1 \right) \, p \left( \lceil \log_2 q \rceil k+1\right) + p\left(1\right) + p\left(p\left(1\right)\right)\right]$ spins and $m\, q\, 2^{k+1} \left[ q\, 2^{k+1}+1\right]$ terms.
  If closure is non-trivial, the overhead is the same model-dependent polynomial $g$ in both these quantities.
\end{theorem}

Note that in a fully general result such as this, it is not meaningful to quantify the reduction in terms of the number of literals and/or clauses of the SAT problem, since this depends on the particular decomposition of the SAT formula into clauses, and the most efficient decomposition will depend on the the particular choice of universal Hamiltonian.
It is also not meaningful to attempt to quantify the closure overhead more precisely than the polynomial function $g$, since this is again Hamiltonian-dependent, whereas the theorem asserts a general result for arbitrary Hamiltonians.
Applying our result to any particular universal model and computing the overhead is a matter of determining the polynomials $p$ and $g$ for that specific model, and plugging them into the theorem.

Note also that NP-completeness of the ground state energy problem does not on its own imply universality. It is true that a spin model whose \textsc{GSE} problem is NP-complete can encode the ground state energy problem of any other spin model, since the latter is in NP. However, universality requires substantially more: we must reproduce the \emph{entire} energy spectrum, not just the ground state energy. Furthermore, we must also reproduce the corresponding spin configurations and partition function.

It is clear that we need some additional structure to the reduction, beyond that required by NP-hardness.
One way to emphasise the difference between a complexity-theoretic polynomial-time reduction and a faithful reduction is to consider what would happen if P$=$NP (which is possible, as far as we know!).
If P$=$NP, then the \textsc{GSE} problem for essentially any Hamiltonian becomes NP-complete, for trivial reasons: we can simply solve the \textsc{SAT} problem instance as part of the polynomial-time reduction (remember we are assuming P$=$NP here!)
to determine whether it is satisfiable or not.
Once we know what the right answer is, we simply construct the trivial constant Hamiltonian $H = 0$ or $H = 1$, depending on whether the \textsc{SAT} instance is satisfiable or not.
If P$=$NP, this is a valid polynomial-time reduction from any \textsc{SAT} instance to the \textsc{GSE} of an extremely trivial Hamiltonian.
But this trivial Hamiltonian clearly cannot be a universal simulator in the sense we require.
This absurd example illustrates that for physical Hamiltonian simulation (as opposed to merely computing ground state energies), simple computational reductions do not capture enough of the structure of the problem.
By proving necessary and sufficient conditions for universal simulation, \cref{thm:universal_characterisation} shows precisely what additional structure is required.

\begin{proof}[of \cref{thm:universal_characterisation}]
  The ``only if'' direction is immediate.
  The closure part follows immediately from the definition of universality (\cref{def:universal}): the model must be able to simulate \emph{any} $H^{(1)}_{A}$, $H^{(2)}_{B}$ and $H = H^{(1)}_{A} + H^{(2)}_{B}$, including any $H^{(1)}_{A}$ and $H^{(2)}_{B}$ that are in the universal model itself.
  The faithful reduction part follows immediately from the fact that the 2D Ising model with fields admits a faithful reduction from \textsc{SAT} (\cref{lem:2DIsing}): one Hamiltonian that must be simulatable by the universal model is the 2D Ising model with fields.

  To prove the ``if'' direction, it will be sufficient to show that we can simulate a single $k$-body Hamiltonian term.
  We will then apply this to each term in the Hamiltonian separately, and use closure of the model to assemble the resulting simulations of the local terms into a simulation of the full Hamiltonian.
  We therefore focus initially on a single $k$-body term $H'$.
  We will first prove that the model satisfies \cref{part:eigenvalues,part:eigenstates} of \cref{def:simulation}, before extending the construction to prove \cref{part:partition_function}.

  To achieve the required separation of energy scales (cf.\ \cref{def:simulation}), we will use the reduction from \textsc{SAT} three times, and combine the resulting Hamiltonians (see \cref{fig:proof-completeness-spins}).
  Each use of the reduction will in general introduce new auxiliary spins (see \cref{def:faithful}).
  We assume without loss of generality that the reduction gives a Hamiltonian whose ground state energy is~0, and all excited states have energy at least~1, as this can always be achieved by shifting and rescaling the energy.

  \begin{figure}[htbp]
    \centering \includegraphics[width=0.6\textwidth]{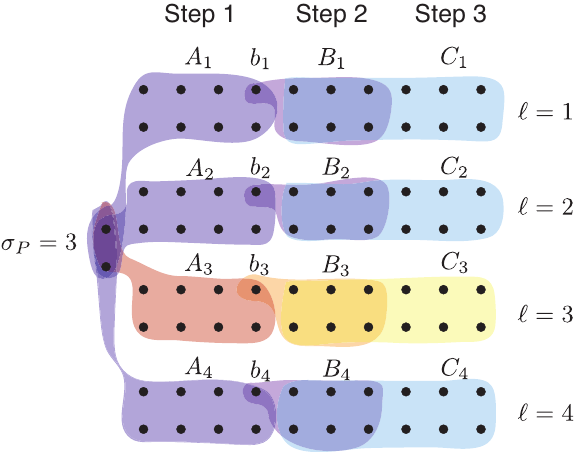}
    \caption{The proof of \cref{thm:universal_characterisation} for the case $k=2$, here for $\sigma_P=3$.
      Step 1 introduces auxiliary spins $A_\ell$ and $b_\ell$, and Steps 2 and 3, $B_\ell$ and $C_\ell$, respectively, for all $\ell$.
      Note that we do not assume any structure in the interactions of the spins.}
    \label{fig:proof-completeness-spins}
  \end{figure}

  We first expand the target Hamiltonian in terms of the elementary basis for Boolean functions (see \cref{sec:prelims}), $H'(x) = \sum_{\ell=1}^{2^{k}} E_\ell e_\ell(x)$ where $E_\ell= H'(\ell)$.

  \paragraph{Step~1}
  \label{step1}
  We apply the reduction to the Boolean formula $\phi_\ell=e_{\ell}(\sigma_P)\xnor\sigma_{b_\ell}$, for $\ell=1,\ldots, 2^{k}$, where $b_\ell$ is a single auxiliary ``flag'' spin.
  This gives a Hamiltonian $H_1^\ell$ acting on spins $P\cup b_\ell\cup A_\ell$ such that
  \begin{equation}\label{eq:H1}
    \begin{cases}
      H_1^\ell(\sigma) = 0 & \sigma_{b_\ell} = e_\ell(\sigma_P), \ \sigma_{A_\ell} \in \Sigma_{A_\ell}^{b_\ell}\\
      H_1^\ell(\sigma) \geq 1 & \text{otherwise},
    \end{cases}
  \end{equation}
  for some non-empty sets of configurations $\Sigma_{A_\ell}^b$ of the auxiliary $A_\ell$ spins.
  That is, the ground state configurations of $H_1^\ell$ satisfy that $\sigma_{b_\ell}= e_\ell(\sigma_P)$, i.e.~the flag spin $\sigma_{b_\ell}$ correctly signals whether $\sigma_P= \ell$ or not.

  \paragraph{Step~2}
  We apply the reduction to the Boolean formula which forces the flag spin to be~0: $\sigma_{b_\ell}\xnor 0$.
  This gives a Hamiltonian $H_2^\ell$ acting on spins $b_\ell\cup B_\ell$ such that
  \begin{equation}\label{eq:H2}
    \begin{cases}
      H_2^\ell(\sigma) = 0 & \sigma_{b_\ell}=0, \ \sigma_{B_\ell}\in \Sigma_{B_{\ell}}\\
      H_2^\ell(\sigma) \geq 1 & \text{otherwise},
    \end{cases}
  \end{equation}
  where $\Sigma_{B_\ell}$ is some non-empty set of configurations of the auxiliary $B_\ell$ spins.
  In particular, choose an arbitrary configuration $\sigma^\star$ and let
  \begin{equation}
    \kappa_\ell := H_2^\ell(\sigma_{b_\ell}=1, \sigma_{B_\ell} = \sigma^\star) \geq 1\, .
  \end{equation}
  Note that computing $\kappa_\ell$ requires computing an energy level of the Hamiltonian $H_2^\ell$.
  Since $H_2^\ell$ is constructed by reduction from a formula involving only a single variable, this is a Hamiltonian on a \emph{constant} number of spins, so $\kappa_\ell$ can be computed in constant time.\footnote{Strictly speaking, in \cref{def:universal} we do not actually require the Hamiltonian to be \emph{constructable} in polynomial-time, only that it has polynomial overhead in terms of the number of spins and parameters. However, our proof also gives this stronger property.}

  \paragraph{Step~3}
  We apply the reduction to the Boolean formula which forces the configuration of $B_\ell$ to be $\sigma_{B_\ell}^\star$ if $\sigma_{b_{\ell}}=1$, and does nothing if $\sigma_{b_{\ell}}=0$:
  \begin{equation}
    \bigl((\sigma_{b_{\ell}}\xnor 1) \land (\sigma_{B_{\ell}}\xnor\sigma_{B_\ell^\star})\bigr) \lor (\sigma_{b_{\ell}}\xnor 0).
  \end{equation}
  This gives Hamiltonian $H_3^\ell$ acting on spins $B_\ell\cup C_\ell$ such that
  \begin{equation}\label{eq:H3}
    \begin{cases}
      H_3^\ell(\sigma) = 0 & \sigma_{b_\ell} = 1,\ \sigma_{B_\ell}=\sigma_{B_\ell}^\star,\ \sigma_{C_\ell} \in \Sigma_{C_\ell}^1\\
      H_3^\ell(\sigma) = 0 & \sigma_{b_\ell} = 0,\ \sigma_{C_\ell} \in \Sigma_{C_\ell}^0\\
      H_3^\ell(\sigma) \geq 1 & \text{otherwise},
    \end{cases}
  \end{equation}
  for some non-empty sets of configurations $\Sigma_{C_\ell}^b$ of the auxiliary $C_\ell$ spins.

  By rescaling these Hamiltonians and adding them together, we construct the following Hamiltonian (which can be simulated by a Hamiltonian within the same model, since the model is closed by assumption):
  \begin{equation}
    \begin{split}
      H(\sigma) = &\Delta \sum_{\ell} H_1^\ell(\sigma_P,\sigma_{b_\ell},\sigma_{A_\ell}) + \sum_{\ell} \frac{E_\ell}{\kappa_\ell} H_2^\ell(\sigma_{b_\ell},\sigma_{B_\ell})\\
      + &\Delta \sum_{\ell} H_3^\ell(\sigma_{B_\ell},\sigma_{C_{\ell}}).
    \end{split}
    \label{eq:Hfinal}
  \end{equation}
  Consider a configuration with
  \begin{align}
    \sigma_P &= x,\\
    \sigma_{A_\ell} &\in \Sigma_{A_{\ell}},\\
    \sigma_{b_\ell} &= e_\ell(x),\\
    \sigma_{B_\ell} &= \sigma_{B_\ell}^\star \text{ for } \ell\neq x,\\
    \sigma_{B_\ell} &\in \Sigma_{B_\ell} \text{ for }\ell=x,\\
    \sigma_{C_{\ell}} &\in \Sigma_{C_{\ell}}^{e_\ell(x)}.
  \end{align}
  (Here $x$ is any particular configuration of the physical spins).
  From \cref{eq:H1,eq:H2,eq:H3}, $H(\sigma)$ evaluated on such a configuration gives energy $E_x$.
  For all other configurations, $H(\sigma)\geq \Delta$.
  So $H$ simulates the low-lying energy levels and configurations of $H'$, as required (see \cref{fig:proof-completeness-step}).
  This proves that the model satisfies \cref{part:eigenvalues,part:eigenstates} of \cref{def:simulation}.

  \begin{figure}[htbp]
    \centering \includegraphics[width=.6\textwidth]{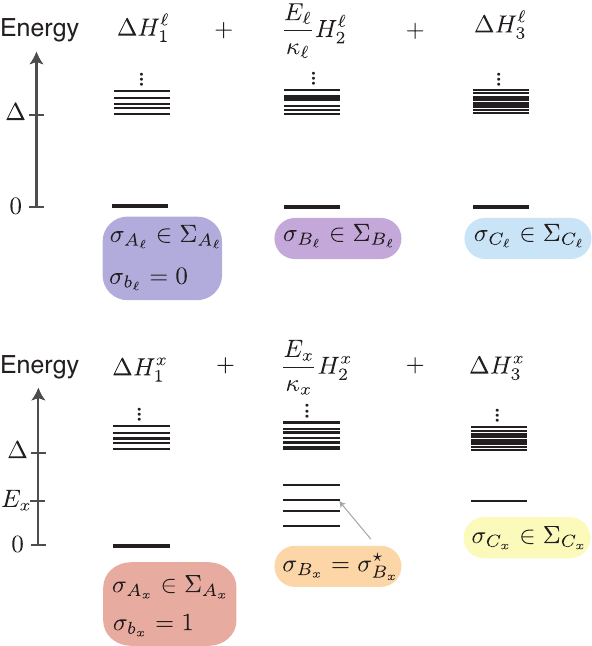}
    \caption{ The proof of \cref{thm:universal_characterisation}.
      The upper side corresponds to $\ell\neq x$, and the lower to $\ell=x$.
      The first, second and third column illustrate the spectrum of $H$ after adding the first, second and third term in \cref{eq:Hfinal}, respectively.
      Note that after the third step the spectrum of $H$ below $\Delta$ coincides with that of $H'$.
      The ground state (or low-lying) configurations are indicated at each step.
    }
    \label{fig:proof-completeness-step}
  \end{figure}

  \paragraph{Partition Function}
  It remains to prove \cref{part:partition_function} of \cref{def:simulation}.
  The previous construction of $H$ simulates the energy levels and corresponding configurations of the target Hamiltonian $H'$.
  However, it does not necessarily reproduce the correct partition function.
  Even though each energy level $E_k$ corresponds to a unique configuration $\sigma_P$ of spins $P$,\footnote{Even if an energy level is degenerate, each of those degenerate energy levels will correspond to a unique $\sigma_P$.} the construction introduces additional auxiliary spins at various steps.
  There may be multiple possible configurations of these auxiliary spins for the same energy $E_k$, and this degeneracy could be different for different energy levels, in which case $H$ will not correctly reproduce the partition function of $H'$.

  Thus, we must modify the construction to ensure that the auxiliary spins introduce the \emph{same} degeneracy for each energy level below the cut-off $\Delta$.
  The idea is to introduce additional auxiliary spins and additional terms in the Hamiltonian, which do not change the energy levels or configurations of the physical spins $P$, but which increase the degeneracy of each energy level in such a way that all the degeneracies are identical.

  \paragraph{Step~1'}
  We use the shorthand notation $\abs{A_\ell}\equiv\abs{\Sigma_{A_\ell}}$.
  From \cref{eq:H1}, the number of zero-energy configurations of $H_1$ with $\sigma_P=x$ is given by
  \begin{equation}
    \left| A_x^1\right|\: \prod_{\ell\neq x}\: \left| A_\ell^0\right| ,
  \end{equation}
  which could depend on $x$.
  By introducing $2^k$ copies of $P$ and symmetrising over all $2^k$ cyclic permutations of the terms $H_1^\ell$ that couple to the flag spins, we can increase this degeneracy to
  \begin{equation}
    \prod_{\ell'} \: \left| A_{\ell'}^1\right|\: \left| A_{\ell'}^0\right|^{2^k-1},
  \end{equation}
  which is independent of $x$.

  More precisely, we introduce $2^k-1$ auxiliary copies $P_{\ell'}$ of the physical spins $P$, and denote $P\equiv P_0$.
  Define the Hamiltonian $H_1^{\ell,\ell'} = H_1^{\ell-\ell'}$ (cf.\ \cref{eq:H1}) acting on spins $P_{\ell'}\cup b_{\ell}\cup A_{\ell,\ell'}$, where the difference $\ell-\ell'$ is taken modulo $2^k$.
  (So in particular $H_1^{\ell,0} = H_1^\ell$.)
  We similarly have $A_{\ell,\ell'} = A_{\ell-\ell'}$ (hence again  $A_{\ell,0} = A_{\ell}$), and $A_{\ell,\ell'}^{e_{\ell}(x)} = A_{\ell-\ell'}^{e_{\ell}(x)}$.
  Construct the overall Hamiltonian
  \begin{equation}
    H_1 = \sum_{\ell,\ell'=1}^{2^k} H_1^{\ell,\ell'}
  \end{equation}
  by combining all of these terms.
  Each flag spin $b_\ell$ is now coupled to $2^k$ copies of $P$, by all possible $2^k$ Hamiltonian terms from \cref{eq:H1} (see \cref{fig:partitionfunction}).

  \begin{figure}[htbp]
    \centering \includegraphics[width=0.85\textwidth]{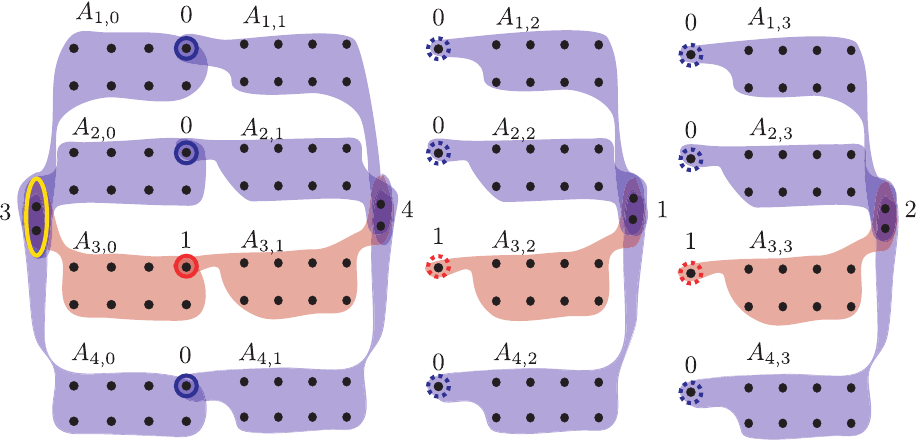}
    \caption{ The modified construction in the proof of \cref{thm:universal_characterisation} to reproduce the partition function.
      Here the configuration of the target Hamiltonian is $x=3$ (encircled in yellow).
      We reuse the flag spins $\sigma_{b_{\ell'}(\ell)}$ as $\sigma_{b_{\pi(\ell')}(\pi(\ell))}$ where $\pi$ is a cyclic permutation of $(1,2,3,4)$.
      The two spins encircled in dashed lines in each row are the same as those to their left.
      This induces a degeneracy of auxiliary spins of the first Hamiltonian which is constant for all $x$.
    }
    \label{fig:partitionfunction}
  \end{figure}

  The $H_1^{\ell,0} $ terms couple the original set of physical spins $P_0$ to the flag spins, by exactly the same Hamiltonians as in our previous construction.
  So, exactly as before, for a given configuration $\sigma_P=x$ of the physical spins, the flag spin $b_\ell$ is forced to take the value $\sigma_{b_\ell}=e_\ell(x)$ (\cref{eq:H1}), and there are $\left| A_{\ell,0}^{e_\ell(x)}\right|$ possible configurations of the auxiliary $A_{\ell,0}$ spins.

  We must show that this is consistent with the constraints imposed by all the new terms $H_1^{\ell,\ell'\neq 0}$.
  Consider the $P_{\ell'}$ copy of $P$.
  It is coupled to flag spin $b_x$ by $H_1^{x,\ell'}$.
  But $\sigma_{b_x}=1$, so by \cref{eq:H1} the only zero-energy configuration of $P_{\ell'}$ is $\sigma_{P_{\ell'}}=x-{\ell'}$.
  All other flag spins $\sigma_{b_{\ell\neq x}}=0$, and $\ell-\ell' \neq x-\ell'$ for $\ell\neq x$, so this configuration satisfies the constraints imposed by all other $H_1^{\ell,\ell'} $.
  Thus there exists a unique zero-energy configuration for each of the other $P_{\ell'}$ copies, with $\left| A_{\ell,\ell'}^{e_{\ell}(x)}\right| $ possible configurations of the corresponding auxiliary  spins $A_{\ell,\ell'}$.

  The configuration $\sigma_P$ of the physical spins uniquely determines $\sigma_{b_\ell}$ for all $\ell$ in any zero-energy configuration.
  Thus the zero-energy configurations of the $A_{\ell,\ell'}$ spins can be chosen independently of one another.
  Since the sets $A_{\ell,\ell'}$ are also completely disjoint, the total degeneracy of the zero-energy configurations with $\sigma_P=x$ is given by the product of all the individual degeneracies:
  \begin{equation}
    \prod_{\ell,\ell'}\: \left| A_{\ell,\ell'}^{e_\ell(x)}\right| \:
    = \: \prod_{\ell'} \: \left| A_{\ell'}^1\right| \: \left| A_{\ell'}^0\right|^{2^k-1},
  \end{equation}
  which is independent of the configuration $\sigma_P=x$ of the physical spins, as claimed.
  $H_1$ will be multiplied by $\Delta$, so all non-zero energy configurations of $H_1$ have energy $\geq\Delta$, and their degeneracies do not concern us.

  \paragraph{Steps~2' and~3'}
  We leave the Hamiltonians $H_2$ and $H_3$ from \cref{eq:H2,eq:H3} unchanged.
  Note that all the terms $H_2^\ell$ from \cref{eq:H2} are the same Hamiltonian, just acting on different subsets of spins $b_\ell\cup B_\ell$ (which implies that all $\kappa_\ell$ in fact take the same value $\kappa$).
  Similarly for the terms $H_3^\ell$ from \cref{eq:H3}, acting on subsets $b_\ell\cup C_\ell$.
  Furthermore, for given $\ell$ and $\ell'$, the terms $H_1^{\ell+N,\ell'+N} = H_1^{\ell,\ell'}$ constructed above are the same Hamiltonian for all $N$, just acting on different subsets of spins $P_{\ell'}\cup b_{\ell}\cup A_{{\ell},{\ell'}}$.
  Let $\pi_N$ be a cyclic permutation with shift $N$ of the $2^k$ blocks $P_{\ell}\cup b_{\ell}\cup A_{\ell}\cup B_{\ell}\cup C_{\ell}$.
  Applying $\pi_N$ to the overall Hamiltonian
  \begin{equation}\label{eq:H}
      H \;= \;\Delta\sum_{\ell,\ell'} H_1^{\ell,\ell'} + \sum_\ell\frac{E_\ell}{\kappa_\ell} H_2^\ell + \Delta\sum_\ell H_3^\ell
  \end{equation}
  transforms it into
  \begin{equation}\label{eq:H_permuted}
    \begin{split}
      \pi_NH \;=& \;\Delta\sum_{\ell,\ell'} H_1^{\ell+N,\ell'+N} + \sum_\ell\frac{E_\ell}{\kappa_\ell} H_2^{\ell+N} + \Delta\sum_\ell H_3^{\ell+N}\\
      \;= &\;\Delta\sum_{\ell,\ell'} H_1^{\ell,\ell'} + \sum_\ell\frac{E_{\ell+N}}{\kappa} H_2^\ell + \Delta\sum_\ell H_3^\ell.
      	\end{split}
  \end{equation}

  Now, from \cref{eq:H2,eq:H3}, we know that a low-energy configuration of $H$ with energy $E_t < \Delta$ must be a zero-energy configuration of $H_1$, $H_3$, and all of the $H_2^\ell$ terms except $H_2^t$.
  Thus if $\sigma$ is a low-energy configuration of $H$ with energy $E_t$, comparing \cref{eq:H,eq:H_permuted} we see that it must also be a low-energy configuration of $\pi_NH$ with energy $E_{t+N}$.
  But $H(\pi_N^{-1}(\sigma)) = \pi_NH(\sigma) = E_{t+N}$, which implies that $\pi_N^{-1}(\sigma)$ is a low-energy configuration of $H$ with energy $E_{t+N}$.
  Since $\pi_N$ is bijective, the configurations with energy $E_t$ are therefore in one-to-one correspondence with the configurations with energy $E_{t+N}$, so these energy levels have the same degeneracy.
  This holds for any cyclic permutation $\pi_N$, thus all energy levels have the same degeneracy, as required.

  We have shown that the Hamiltonian $H$ from our modified construction \cref{eq:H} reproduces the low-lying energy levels and configurations of any target Hamiltonian, and that the degeneracy $\gamma$ is the same for all low-lying energy levels $E_t < \Delta$.
  Thus, not only does it satisfy \cref{part:eigenvalues,part:eigenstates} of \cref{def:simulation}, as before, but also
  \begin{equation}
    Z_H(\beta) = \sum_\sigma e^{-\beta H(\sigma)}
    = \gamma\sum_{E_t<\Delta} e^{-\beta E_t} + \sum_{E_{t}\geq\Delta}e^{-\beta E_t}
    = \gamma Z_{H'} + O(e^{-\Delta}).
  \end{equation}
  Thus $H$ satifies \cref{part:partition_function}, too.
  Since $H$ is a sum of Hamiltonians $H_1^{\ell,\ell'}$, $H_2^\ell$ and $H_3^{\ell'}$ from the model, by the closure assumption (\cref{def:closure}) it can itself be simulated by a Hamiltonian from the model.

  So far, we have simulated a single $k$-body Hamiltonian $H'$, using a Hamiltonian $H$ on $n+O(2^k)$ spins with $2^k$ parameters $E_\ell$.
  To simulate an arbitrary  target Hamiltonian $H' = \sum_{I=1}^m h_I$ made up of $m$ separate $k$-body terms $h_I$, we simply use the above construction for each $h_I$ individually.
  The resulting Hamiltonian
  \begin{equation}
    H = \sum_{I=1}^m \sum_{\ell,\ell'=1}^{2^k} \left( \Delta H_{I,1}^{\ell,\ell'} + \frac{E_{I,\ell}}{\kappa_\ell} H_{I,2}^\ell + \Delta H_{I,3}^\ell\right)
  \end{equation}
  Again by closure, $H$ can be simulated by a Hamiltonian from the model.

  To complete the proof of \cref{thm:universal_characterisation}, we must account for the overhead incurred in the construction.
  Consider a target Hamiltonian $H'$ on $n$ $q$-level spins, with $m$ $k$-body terms.
  We can always re-express $H'$ as an equivalent model on 2-level spins by mapping the $q$ levels to the states of $\lceil \log_{2}q\rceil$ 2-level spins.
  This increases the interactions from $k$-body to $k'$-body, where $k'=k\, \lceil \log_{2}q\rceil$.

  In our construction of a Hamiltonian $H$ that reproduces the energy levels and configurations, for a fixed interaction term and fixed $\ell$, the overheads in Steps~1, 2 and~3 are $p(k'+1), p(1)$ and $p(p(1))$, respectively.
  Summing over all $\ell$ and all interaction terms, the number of variables is at most
  \begin{equation}
  m\, 2^{k'}\left[ p\left(k'+1\right) + p\left(1\right)+p\left(p\left(1\right)\right)\right].
  \end{equation}
  On the other hand, the number of Hamiltonian terms in $H$ is at most $3\, m\, 2^{k'}$.

  Finally, $H$ simulates the sum of all these terms (\cref{eq:Hfinal}) by applying the closure property.
  If closure is trivial, e.g.\ because there are no constraints on the locality of interactions, then this incurs no overhead.
  Otherwise, it incurs polynomial overhead
  \begin{equation}
    g\left(m \, 2^{k'} \left[ p\left(k'+1\right) + p\left(1\right)+p\left(p\left(1\right)\right)\right], 3\, m\, 2^{k'}\right).
  \end{equation}

  A similar calculation shows that the number of variables involved in the construction of the partition function simulation is at most
  \begin{equation}
    m \, 2^{k'}\left[ \left(2^{k'}-1\right) p\left(k'+1\right)+p\left(1\right)+p\left(p\left(1\right)\right)\right],
  \end{equation}
  the number of terms in the universal model is at most $m \, 2^{k'}\left( 2^{k'}+1 \right)$, and the closure overhead is again given by the polynomial $g$ of this number of variables and terms.
\end{proof}

In the above proof,  note that we do \emph{not} need to know the energy levels and configurations of the overall Hamiltonian $H'$ in order to simulate it. It is sufficient to know the energy levels and spin configurations of the individual terms $h'_I$ in the Hamiltonian $H' = \sum_I h'_I$, and these can always be computed efficiently.
Note also that the proof is constructive.

\begin{remark}[Parsimonious reductions]
  \label{rem:parsimonious}
  Steps 1' to 3' in the proof of \cref{thm:universal_characterisation} are not needed if the faithful \textsc{SAT} reduction for a particular model has an additional property.
  A \emph{parsimonious} reduction between two computational problems is a reduction that preserves the number of solutions \cite{Mo11,Pa95}.
  A \emph{$c$-monious} reduction is a reduction such that the number of solutions of every instance is multiplied by some constant factor $c$ (independent of the instance).
  If there is a $c$-monious reduction from \textsc{SAT} to the \textsc{GSE} of the model we are interested in, then Steps 1 to 3 in the proof of \cref{thm:universal_characterisation} already guarantee the partition function of the target model is reproduced up to a constant factor $\gamma$, which is usually an easy-to-compute function of $c$ (see \cref{def:simulation}).
\end{remark}
We will see examples of models that admit $c$-monious reductions in \cref{sec:examples}.

\cref{thm:universal_characterisation} and \cref{lem:2DIsing} imply that one of the simplest  2D models, namely the 2D Ising model with fields (\cref{def:2DIsing}), is universal.

\begin{corollary}
  The 2D Ising model with fields is universal.
  \label{cor:2DIsingfields}
\end{corollary}

We remark that \cref{thm:universal_characterisation} result implies and explains the recently found ``completeness results'', where a model is called ``complete'' if its partition function can equal (up to a factor) the partition function of any other model~\cite{Va08}.
Any universal model is complete by choosing $\Delta=\infty$.
The 2D Ising model with fields for \emph{imaginary} coupling strengths and fields~\cite{Va08,Ka12b}, the 3D Ising model for a restricted class of models~\cite{De09a}, and the 4D Ising lattice gauge theory~\cite{De09b} were all shown to be complete in this sense, and similar results were found for $\phi^{4}$ theories~\cite{Ka12} and models with continuous variables~\cite{Xu11}.
Our results also show that the 2D Ising model with physical \emph{real}-valued coupling strengths and fields is universal (\cref{cor:2DIsingfields}).

\section{Universal Hamiltonian Simulation -- continuous spins}\label{sec:continuous}
We now show that the results of \cref{thm:universal_characterisation} extend to models with continuous spin degrees of freedom.
Nothing in the proof of \cref{thm:universal_characterisation} (other than the notation) depends on the universal model having discrete degrees of freedom.
The only requirement is that the universal model has a faithful reduction from SAT, which implies that there is an identification between boolean values and particular values of the continuous degrees of freedom, and that the model is closed.
\cref{thm:universal_characterisation} therefore extends immediately to universal models with continuous spins (or more general continuous local degrees of freedom).

It remains to extend \cref{thm:universal_characterisation} to \emph{target} models with continuous spin degrees of freedom, i.e.\ where the spins are normalised vectors in $\mathbb{R}^D$, or equivalently points on the surface of the unit $D$-dimensional sphere, $S^D$.
Of course, no model with discrete degrees of freedom can exactly simulate a model with continuous degrees of freedom.
We need to show that the energy levels and configurations of the continuous model can be approximated to any desired accuracy by the universal model, implying that even universal models on Ising spins can simulate (in the strong sense of \cref{def:simulation}) models with continuous variables to any desired precision.
The following extends \cref{def:simulation} of Hamiltonian simulation to models with continuous spin degrees of freedom.
Similarly as above, we denote the spin configuration by $s=(s_1,\dots,s_n)$, where $s_{i}\in S^{D}$.

\begin{definition}[Hamiltonian Simulation -- continuous case]
\label{def:continuous_simulation}
  We say that a spin model can \emph{simulate} a Hamiltonian $H'$ on continuous spins $s_i\in S^D$ if it satisfies \emph{all three} of the following:
  \begin{enumerate}
  \item For any $\Delta > \max_{s'} H'(s')$ and any $0 < \delta < 1$, there exists a Hamiltonian $H$ from the model whose low-lying energy levels $\{E_\sigma=H(\sigma):E_\sigma<\Delta\}$ approximate the energy levels $\{E'_{s'}=H'(s')\}$ of $H'$ to within additive error at most $\delta$.
    \label[part]{part:eigenvalues_cont}
  \item There is a one-to-one correspondence $F(\sigma_{P_i}) = s'_i$ between spin-configurations $\sigma_{P_i}$ of fixed subsets $P_i$ of spins in $H$, and the states $s'_i$ of each spin in $H'$, such that for all configurations $\sigma$ with energy $E_\sigma<\Delta$, the energy $E'_{s'} = H'(s')$ with $s'_i = F(\sigma_{P_i})$ satisfies $\abs{E'_{s'} - E_\sigma} \leq \delta$.
    \label[part]{part:eigenstates_cont}
  \item
    The partition function
    $Z_H(\beta) = \sum_\sigma e^{-\beta H(\sigma)}$ of $H$ reproduces the partition function
    $Z_{H'}(\beta) = \sum_{s'} e^{-\beta H'(s')}$ of $H'$ up to rescaling, to within arbitrarily small error: $Z_{H'}(\beta) = \gamma(1+\delta)Z_{H}(\beta) + O(e^{-\Delta})$ for some known constant $\gamma$.
    \label[part]{part:partition_function_cont}
  \end{enumerate}
\end{definition}
As before, we denote the set of physical spins by $P=\cup_{i} P_{i}$. Its cardinality is again independent of $\Delta$, and scales as $\poly(1/\delta)$.

Similar to the proof of \cref{thm:universal_characterisation}, we first prove that a single $k$-body Hamiltonian term on continuous spins can be approximated by a Hamiltonian on $q$-level Ising spins.
Note that the resulting Ising spin Hamiltonian is then simulatable (in the sense of \cref{def:simulation}) by a universal model, by virtue of \cref{thm:universal_characterisation}.
The result for an arbitrary target Hamiltonian made up of multiple $k$-body terms on continuous spins then follows easily, by approximating each term individually, and using closure of the universal model to simulate the sum of all these individual terms.
The main difference between this and the discrete case is that we must carefully keep track of the approximation error throughout.

\begin{lemma}\label{lem:continuous_term}
  Let $h'(s_1,\dots,s_k)$ be a $k$-body Hamiltonian term on continuous spins $s_i\in S^D$ which is Lipschitz-continuous in each argument, i.e.\ $\abs{h'(s_1,\dots,s_i+\epsilon,\dots,s_k) - h'(s_1,\dots,s_i,\dots,s_k)} \leq L\epsilon$ for any $i$, where $L$ is the Lipschitz constant.
  Then $h'$ can be simulated to accuracy $kL\epsilon$ by a $k$-body Hamiltonian term on $q$-level Ising spins, where $q = \left(\tfrac{2}{\epsilon}+1\right)^D = O\left(1/\epsilon^D\right)$.
\end{lemma}

\begin{proof}
  To approximate the continuous spins by discrete Ising spins, we discretise them in the obvious way.
  Take an $\epsilon$-net for the space of the spins, i.e.\ a set of points $\mathcal{P} = \left\{ p_i \in S^D\right\}_{i=1}^q$ such that any point $s \in S^D$ is $\epsilon$-close to some point $p_i$ in the net: $||s-p_i|| \leq \epsilon$, where $||y||$ is the ``maximum norm'' for vectors, $||y|| := \max_j \left|y_j\right|$.
  A standard argument shows that it suffices to take a net containing $q = O(1/\epsilon^D)$ points.

  As this argument is straightforward, we repeat it here to make the proof self-contained.
  Assume that $\mathcal{P}$ is an $\epsilon$-net with cardinality $q$.
  Since $p_i\in S^{D}$, the union of the balls of radius $\epsilon/2$ around $p_i$ is contained in the ball of radius $1+\epsilon/2$ centered at the origin.
  Thus
  \begin{equation}
    q \left(\epsilon/2\right)^D \leq \left(1+\epsilon/2\right)^D,
  \end{equation}
  and $q \leq \left(\tfrac{2}{\epsilon}+1\right)^D = O\left(1/\epsilon^D\right)$.

  We discretise each continous spin $s_i$ by approximating its value by the closest point $p_{\sigma_i}$ in the $\epsilon$-net, and identify each point $p_{\sigma_i}$ in the net with a distinct value $\sigma_i$ of a $q$-level Ising spin.
  Thus a configuration $(s_1,\dots,s_k)$ of the $k$ continuous spins is approximated by the configuration $(\sigma_1,\dots,\sigma_k)$ of the $k$ $q$-level Ising spins.

  Define the $k$-body Hamiltonian term $h$ on $q$-level Ising spins to take the same value as $h'$ on the corresponding points in the $\epsilon$-net: $h(\sigma_1,\dots,\sigma_k) = h'(p_{\sigma_1},\dots,p_{\sigma_k})$.
  Since $h'$ is Lipschitz-continuous in each argument, we have
  \begin{equation}
    \Abs{ h(\sigma_1,\dots,\sigma_k) - h'(s_1,\dots,s_k) }
    = \Abs{ h'(p_{\sigma_1},\dots,p_{\sigma_k}) - h'(s_1,\dots,s_k) }
    \leq kL\epsilon,
  \end{equation}
  which completes the proof.
\end{proof}

\Cref{thm:universal_characterisation} can now readily be extended to show that any universal model can also simulate (in the sense of \cref{def:continuous_simulation}) all models on continuous spin degrees of freedom.

\begin{corollary}
  \label{cor:cont}
  Any Hamiltonian $H'$ on continuous spin degrees of freedom which is Lipschitz-continuous can be simulated (in the sense of \cref{def:continuous_simulation}) by a universal model $H$.

  If $H'$ is a Hamiltonian on $n$ continuous $D$-dimensional spins with $m$ $k$-body terms, each of which is Lipschitz-continuous in each argument with constant $L$ (cf.\ \cref{lem:continuous_term}), then the overhead for approximating the energy levels and configurations to accuracy $\delta$ is at most the overhead given in \cref{thm:universal_characterisation} for \mbox{$q = \left(\tfrac{2\, m\, k\, L}{\delta}+1\right)^D$}.

  Approximating the partition function to accuracy $\delta$ (see \cref{def:continuous_simulation}) increases this to at most the overhead given in \cref{thm:universal_characterisation} for $q = \left(\tfrac{2\, \beta\, m\, k\, L}{\delta}+1\right)^D$.
\end{corollary}

\begin{proof}
  That the energy levels and corresponding configurations can be simulated follows almost immediately from \cref{lem:continuous_term,thm:universal_characterisation}.
  If we use \cref{lem:continuous_term} to approximate each term in $H'$ to accuracy $k\, L\, \epsilon$, and sum over all these terms to obtain a Hamiltonian $H$ on $q$-level Ising spins, we approximate the overall Hamiltonian $H'$ to accuracy $m\, k\, L\, \epsilon$.
  Thus to achieve an accuracy $\delta$, we need $\epsilon = \delta/(m\, k\, L)$, hence $q \leq \left(\tfrac{2\, m\, k\, L}{\delta}+1\right)^D = O\left(\left(m\, k\, L/\delta\right)^D\right)$.
  But this Ising spin Hamiltonian can be simulated by the universal model, by virtue of \cref{thm:universal_characterisation}, and we are done.

  To show that this also approximates the partition function up to an overall rescaling (as in \cref{def:continuous_simulation}), we must compute the approximation error in the partition function.
  Note that we are comparing the partition function of a discrete model $H$ -- given by a sum over $q^n$ configurations -- to that of a continuous model $H'$ -- given by an integral over the space $(S^D)^{\times n}$.
  Thus the partition function of the discrete model will need to be rescaled by $1/q^{n\, D}$ to account for the volume element of the continuous space.

  Taking this rescaling into account in the definition of $Z_H(\beta) = \sum_\sigma e^{-\beta \, H(\sigma)}/q^{n\, D}$ and using Lipschitz-continuity of $H'$, we obtain
  \begin{equation}
    e^{-\beta \, m\, k\, L\, \epsilon} \, Z_H(\beta) \, \leq \, Z_{H'}(\beta)\, \leq\, e^{\beta\, m\, k\, L\, \epsilon} \, Z_H(\beta).
  \end{equation}
  Thus to achieve the multiplicative error $\delta$ in \cref{def:continuous_simulation}, we need $e^{\beta\, m\, k\, L\, \epsilon}-1\leq\delta$.
  For $\delta < 1$ it suffices to take $\epsilon \leq \delta/(\beta \, m\, k\, L)$, and the claimed value of $q$ follows from \cref{lem:continuous_term}.

  Once again, having approximated the continuous Hamiltonian $H'$ by a Hamiltonian $H$ on $q$-level Ising spins, the latter can now be simulated exactly by the universal model, by virtue of \cref{thm:universal_characterisation}.
\end{proof}

\section{Examples}\label{sec:examples}
In this section, we give various examples of how discrete models can be simulated using the Ising model. We also present alternative and simpler reductions from \textsc{SAT} to the \textsc{GSE} of other types of Ising model,
which are useful in showing that other important and well studied models are also universal.

\begin{definition}[The 3D Ising model]
  \label{def:3DIsing}
  The ``3D Ising model'' is defined as a subfamily of the Ising model with fields (\cref{def:Isingmodel}) in which $G$ is restricted to be a three-dimensional (3D) square lattice, and there are no local fields (i.e.~$h_i=0$ for all $i$).
\end{definition}

Note that ``the 3D Ising model'' is defined without local fields.
We could prove directly that the 3D Ising model is universal by showing that it can simulate the 2D Ising model with fields, which is universal.
But by introducing yet another NP-complete computational problem, we can prove this in a simpler and more direct way.
\begin{definition}[\textsc{1-in-3-SAT}]
  \label{1in3SAT}
  Given a set of ``clauses'' each containing three boolean variables or their negations, does there exist an assignment to the variables in which \emph{exactly} one variable per clause is true?
\end{definition}
Clauses in a \textsc{1-in-3-SAT} problem will be called \emph{1-in-3-SAT clauses}. Note the clauses are merely subsets of the variables, not disjunctions.

\begin{example}
  The 3D Ising model is universal.
  \label{ex:3DIsing}
\end{example}

\begin{proof}[of \cref{ex:3DIsing}]
By \cref{thm:universal_characterisation}, we only need to show that the \textsc{GSE} of the 3D Ising model (\cref{def:3DIsing}) admits a faithful reduction from \textsc{SAT}, and that the model is closed.
To see the former, we proceed as follows:

\begin{enumerate}
\item
  \emph{Express the Boolean formula in 3 Conjunctive Normal Form (3CNF)},
  exactly as explained in \cref{3CNF} in the proof of \cref{lem:2DIsing}.

\item
  \emph{Reduction from \textsc{3SAT} to \textsc{1-in-3-SAT}}.
  We use the standard reduction \cite{Sc78}.
  Each \textsc{3SAT} clause $\chi:=(x\lor y\lor z)$ is transformed into three \textsc{1-in-3-SAT} clauses:
  \begin{equation}
    \tilde\chi := \mathcal{R}(\bar x,a,b) \land \mathcal{R}(y,b,c) \land\mathcal{R}(\bar z, c, d),
  \end{equation}
  where we have introduced four additional boolean variables, $a,b,c,d$, and $\mathcal{R}(x,t,z)$ denotes the boolean function which is 1 iff exactly one of the variables $x$, $y$ or $z$ is 1.
  It is easy to verify that $\chi$ is satisfiable in \textsc{3SAT} iff $\tilde{\chi}$ is satisfiable in \textsc{1-in-3-SAT}.
  Moreover, for all but one of the satisfying assignments of $\chi$, there exists a unique assignment to $(a,b,c,d)$ such that $\tilde\chi$ can be satisfied.
  (The only exception is $(x,y,z)=(1,0,1)$, for which both $(a,b,c,d)=(0,1,0,1)$ and $(1,0,1,0)$ satisfy $\tilde{\chi}$.)

\item
  \emph{Reduction from \textsc{1-in-3SAT} to \textsc{GSE} of Ising with fields}.
  We transform an instance $\tilde{\phi}$ of \textsc{1-in-3SAT} to an Ising model with fields (\cref{def:Isingmodel}) on a graph $G=(V,E)$ and with a parameter set $\mathcal{J}$ as follows.
  First, for every boolean variable $x_i$ in $\tilde{\phi}$, we create a pair of vertices in $G$ labelled $x_i$ and $\bar{x}_{i}$, associated to the Ising variables $\sigma_i$, $\sigma_{\bar i}$, respectively, and construct the Ising term
  \begin{equation}
    h_i = \sigma_i \sigma_{\bar{i}} + 1 .
  \end{equation}
  Second, for each 1-in-3-SAT clause $c$ in $\tilde{\phi}$, define the Ising term
  \begin{equation}
    h_c = \sum_{i<j} \sigma_i \sigma_j + \sum_i \sigma_i + 2 ,
  \end{equation}
  where $\{\sigma_i\}$ are the Ising variables corresponding to the Boolean variables $\{x_i\}$ that appear in $c$.
  Finally, wherever we have a Boolean variable in a fixed state (0 or 1), we use the Hamiltonian
  \begin{equation}
    h_f = 1 +\sigma_f \quad \textrm{or}\quad h_f= 1 - \sigma_f ,
  \end{equation}
  respectively, where $f$ indexes the spin representing the variable whose value is fixed.
  The full Hamiltonian is then
  \begin{equation}
    H_{G,\mathcal{J}} = \sum_{\textrm{clauses }c} h_c + \sum_{\textrm{variables }i} h_i + \sum_{\textrm{fixed variables }f} h_f + K
    \label{eq:HGJ}
  \end{equation}
  where $K=2 (\tilde{m} + p)$, where $\tilde{m}$ is the number of clauses of $\tilde{\phi}$ and $p$ the number of variables in $\tilde{\phi}$ such that the variable and its negation appear in $\tilde \phi$.
  It is easy to see that $\tilde{\phi}$ is satisfiable in 1-in-3SAT iff $H_{G,\mathcal{J}}$ has a ground state with energy 0.
  Moreover, the state of the Ising spin $\sigma_j\in\{-1,1\}$ associated to the Boolean variable $x_j\in\{0,1\}$ is determined by
  \begin{equation}
    \sigma_j = 2 x_j -1 .
    \label{eq:Boolean-Ising}
  \end{equation}

  For example, the satisfiability of  $\phi_{1}$,
  \begin{equation}
  \phi_{1}=(x_{1}\lor x_{2} \lor \sigma_{b_{1}}) \land
(x_{1}\lor \bar x_{2} \lor \bar \sigma_{b_{1}}) \land
(\bar x_{1}\lor x_{2} \lor \bar \sigma_{b_{1}}) \land
(\bar x_{1}\lor \bar x_{2} \lor \bar \sigma_{b_{1}}) ,
\label{eq:phi1k2}
  \end{equation}
  in \textsc{3SAT} is mapped to the satisfiability of $\tilde{\phi}_1$,
   \begin{align}
     \tilde{\phi}_{1} =&
     \mathcal{R}(\bar x_{1} ,a_{1} ,b_{1} ) \land \mathcal{R}(x_{2} , b_{1} , c_{1} ) \land \mathcal{R}(\bar \sigma_{b_{1}} , c_{1} , d_{1} ) \land \nonumber\\
                      & \mathcal{R}(x_{1} , a_{2} , b_{2} ) \land  \mathcal{R}(x_{2} , b_{2} , c_{2} ) \land  \mathcal{R}(\sigma_{b_{1}} , c_{2} , d_{2} ) \land \nonumber\\
                      & \mathcal{R}(\bar x_{1} , a_{3} , b_{3} ) \land  \mathcal{R}(\bar x_{2} , b_{3} , c_{3} ) \land  \mathcal{R}(\sigma_{b_{1}} , c_{3} , d_{3} ) \land \nonumber\\
                     & \mathcal{R}(x_{1} , a_{4} , b_{4} ) \land  \mathcal{R}(\bar x_{2} , b_{4} , c_{4} ) \land  \mathcal{R}(\sigma_{b_{1}} , c_{4} , d_{4} ) ,
                       \label{eq:tildephi1}
  \end{align}
  in \textsc{1-in-3SAT}, which is mapped to an Ising model on the graph shown in \cref{fig:1in3SAT}.

  \begin{figure}[t]
    \centering \includegraphics[width=0.5\textwidth]{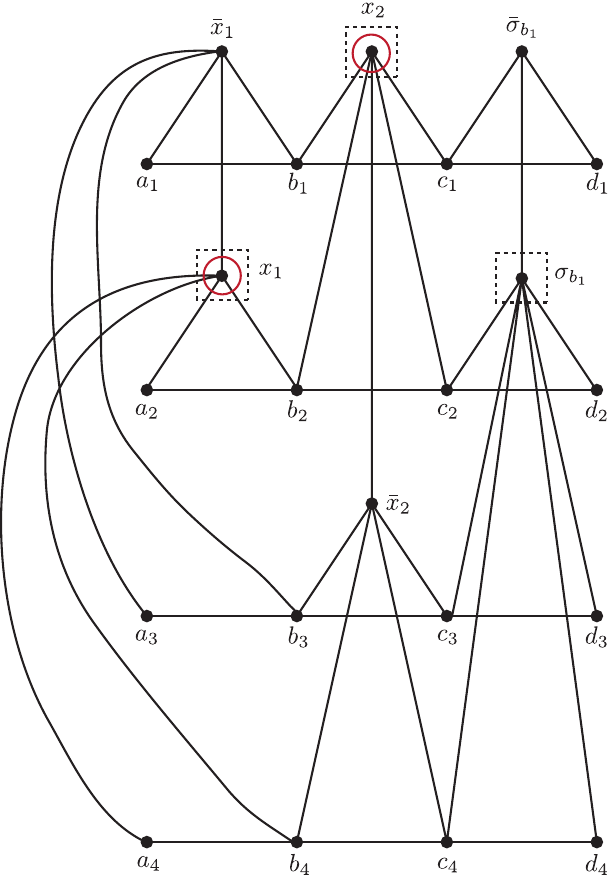}
    \caption{The formula $\tilde{\phi}_1$ (\cref{eq:tildephi1}) is satisfiable in \textsc{1-in-3SAT} iff the Ising Hamiltonian \cref{eq:HGJ} defined on the graph shown here has a ground state of energy 0.
            The variables belonging to $R$ (see \cref{def:faithful}) are marked with dashed squares, and those in the physical set $P$ (see \cref{def:simulation}) are marked with red circles.}
    \label{fig:1in3SAT}
  \end{figure}

  Let $\tilde{n}$ and $\tilde{m}$ denote the number of variables and clauses of $\tilde{\phi}$, respectively.
  Then this construction yields a graph $G$ with $\tilde{n}$ vertices and $3\tilde m +p$ edges.
  The relation with $\phi$, the 3CNF expression of $\tilde{\phi}$, is also straightforward: we have $\tilde{n} = n+4m$ and $\tilde{m} = 3m$, where $n$ and $m$ denote the number of variables and clauses in $\phi$, respectively.

  Now consider the sequence of transformations \textsc{SAT} $\to$ \textsc{3SAT} $\to$ \textsc{1-in-3SAT} $\to$ \textsc{GSE} of Ising with fields that we have applied to a formula $\phi_j$ depending on boolean variables  $(x_1,\ldots, x_k,b_{j})$
  (such as $\phi_{1}$ for $k=2$ and the first flag spin,  see \cref{eq:phi1k2}).
  Each of the variables in $\phi_j$ is mapped to a single spin in \textsc{GSE} of Ising with fields, which form the set $R$ of spins, $\sigma_R=(\sigma_1,\sigma_2,\ldots, \sigma_k, \sigma_{b_j})$.
  The rest of spins in $H$ are auxiliary, which we denote with the subscript $\textrm{aux}$.
  Then for each configuration of $\sigma_R$ for which there exists some state  $\sigma_{\textrm{aux}}$ such that $(\sigma_R,\sigma_\textrm{aux})$ is a ground state, the Boolean variable assignment corresponding to $\sigma_R$ is satisfying.
  Hence the reduction is faithful.
  Note that the set $P$ of physical spins is $\sigma_P=(\sigma_1,\sigma_2,\ldots, \sigma_k)$.

\item
  \emph{Reduction from \textsc{GSE} of Ising with fields to \textsc{GSE}  of 3D Ising model}.
  \label{Ising->3DIsing}
  First we reduce the \textsc{GSE} of the Ising model with fields on a graph $G$ to the
  \textsc{GSE} of the Ising model without fields on a graph $G'$.
Consider the Ising Hamiltonian with fields of \cref{eq:HGJ}, which can be written as
    \begin{equation}
      H_{G} ( \sigma ) = \sum_{(i,j)\in E} J_{ij} s_i s_j + \sum_{i\in V} h_i \sigma_i + K'\, .
    \end{equation}
    for some parameters $\{J_{i,j}\},\{h_{i}\}$ and $K'$.
    We introduce one additional spin, $\sigma_{\textrm{spike}}$, and define the following Ising Hamiltonian without fields,
    \begin{equation}
      H_{G'}( \sigma , \sigma_{\textrm{spike}})= \sum_{(i,j)\in E} J_{ij} \sigma_i \sigma_j
      + \sum_{i\in V} h_i \: \sigma_i \sigma_{\textrm{spike}}
      + \Delta \frac{1-\sigma_{\textrm{spike}}}{2} + K' .
    \end{equation}
Now, $H_{G'}$ simulates the sector with energy $<\Delta$ of $H_{G}$, since  $\sigma_{\textrm{spike}}=1$ for all configurations with energy $<\Delta$.

    Finally, we simulate $H_{G'}$ with an Ising Hamiltonian without fields on a 3D square lattice.
    To do this, we embed $G'$ in a 3D square lattice (which only needs to be polynomially larger \cite{Co94}), and obtain $G'$ as a minor of this square lattice by edge contractions and deletions.
    The edge contraction operation is accomplished by setting $J_{i,j}>\Delta$, and edge deletion by setting $J_{i,j}=0$.
    (This is the same idea as \cref{IsingG->2D} in the proof of \cref{lem:2DIsing}, but here both $G'$ and the 3D square lattice are non-planar graphs, and we do not have local fields.)
    The resulting 3D Ising model without fields will simulate $H_{G'}$.

\end{enumerate}

Finally, the 3D Ising model is closed because given any two 3D square lattices $G$ and $G'$, $G\cup G'$ can be embedded in 3D space and obtained as the minor (again using the techniques of \cref{IsingG->2D} in the proof of \cref{lem:2DIsing}) of some other, polynomially larger 3D square lattice $G''$ \cite{Co94}.
\end{proof}

We now present another example that illustrates in a more direct way the idea of encoding a Boolean formula into the ground state, and how we use it to localise information about the spin configuration in single ``flag'' spin.
At the same time it illustrates how specific features of the model can be exploited to simplify the general construction.

\begin{definition}[Potts model]
  \label{def:potts}
  The Potts model on an arbitrary graph $G=(V,E)$ with inhomogeneous couplings $J_{i,j}\in \mathbb{R}$ is defined as the set of Hamiltonians of the form
  \begin{equation}
    H_{\textrm{Potts}, G}(\sigma)= \sum_{(i,j)\in E} J_{i,j}\: \delta(\sigma_i,\sigma_j)
    \label{eq:potts}
  \end{equation}
  where $J_{i,j}\in \mathbb{R}$ and $\sigma_{i}\in \{1,\ldots,q\}$. $\delta(\sigma_i,\sigma_j)=1$ if $\sigma_i=\sigma_j$, otherwise $\delta(\sigma_i,\sigma_j)=0$.
\end{definition}
The Potts model is a generalisation of the Ising model \cite{Wu84}; another possible generalisation is the clock model (\cref{def:clock}), which we will consider later.

\begin{example}[Potts model to Ising model with fields]
  \label{ex:Potts}
  We consider the Potts model with $q=4$ (see \cref{def:potts}) defined on a graph $G'=(V',E')$,
  \begin{equation}
    H'_{G'} (\sigma')= \sum_{(i,j)\in E'} J_{i,j} \delta(\sigma'_i,\sigma'_j)
    \label{eq:PottsG'}
  \end{equation}
  with $\sigma'_i\in\{1,2,3,4\}$.
  We simulate this model with the Ising model with fields (\cref{def:Isingmodel}) on a graph $G=(V,E)$.
  We will follow the Steps in the proof of \cref{thm:universal_characterisation}, but will adapt them to this specific model.

  We first consider a single interaction of $H'_{G'}$ of edge $e$ between, say, $\sigma'_{1}$ and $\sigma'_{2}$,
  \begin{equation}
    J_{1,2}\: \delta(\sigma'_1,\sigma'_2) .
    \label{eq:potts-1term}
  \end{equation}
  We will later use closure of the Ising model with fields to simulate $H'_{G'}$.

We map each $\sigma'_{i}$ to $\lceil \log_{2} 4\rceil = 2$ binary variables $x_j\in \{0,1\}$: $(x_1,x_2)$ will be the binary expression of $\sigma_1$, and similar for $(x_3,x_4)$ and $\sigma_2$.

  \begin{enumerate}
  \item \label[step]{im:potts-formula}
    Now we present an abridged version of Step 1 in the proof of \cref{thm:universal_characterisation}.
    Since the interaction term \eqref{eq:potts-1term} has only one non-zero energy, we construct only one flag spin $\sigma_b$. We want to construct a hamiltonian such that  in its ground state, $\sigma_b= 1$  only if $\sigma'_1=\sigma'_2$  and $\sigma_b= -1$ for any other configuration of $\sigma'_1,\sigma'_2$.
    To this end, given Boolean variables $u,v,w$, consider the formula
    \begin{equation}
      \phi(u,v,w):= (u \lor v\lor w) \land (u \lor \bar v\lor \bar w) \land (\bar u \lor v\lor \bar w) \land (\bar u \lor \bar v\lor w) .
      \label{eq:phi}
    \end{equation}
    The satisfying assignments of $\phi$ verify that $w=1$ iff $u=v$.
    Now, let
    \begin{equation}
      \Phi:=\phi(x_1,x_3,g_1) \land \phi(x_2,x_4,g_2) \land \phi(g_1,g_2,b) ,
      \label{eq:Phi}
    \end{equation}
    where we omit the dependence of $\Phi$ on the variables for clarity.
    The satisfying assignments of the first and second clause are such that $g_1=1$ iff  $x_1=x_3$, and $g_{2}=1$ iff $x_2 =x_4$, respectively.
    The third clause is satisfied if $b=1$ iff $g_{1}=g_{2}$.
    This is nearly what we want, except for the fact that  $b=1$ for $g_{1}=g_{2}=0$ is also a valid satisfying assignment. We will increase the energy of this last configuration in the next step, and thereby solve this problem.

    (This could of course be directly solved by replacing the last clause  in \eqref{eq:Phi} by a formula which is satisfied if $b=1$ iff $g_{1}=g_{2}=1$.
    The reduction of this formula to the Ising model with fields, however, introduces different number of ground state configurations for each satisfying assignment, which would make the reproduction of the partition function (\cref{im:Z-potts} below) less straightforward. The reduction of \eqref{eq:Phi}, in contrast, will give us an example of a $c$-monious reduction (see \cref{rem:parsimonious})).

  \item
    \label{im:sat-ising-potts}
    We now reduce the satisfiability of $\Phi$ to the \textsc{GSE} of the Ising model with fields.
    We first reduce \textsc{3SAT} to \textsc{Vertex Cover} (\cref{im:3SAT-VC} in the proof of \cref{lem:2DIsing}), and then \textsc{Vertex Cover} to \textsc{GSE} of Ising with fields (as in \cref{im:PVC-Ising} in the proof of \cref{lem:2DIsing}; this reduction is the same irrespective of whether the initial graph is planar or not).

    For the formula $\phi$ this results in the Ising model on the graph $G_\phi$ shown in the top left of \cref{fig:Potts}, with the Hamiltonian defined in \eqref{eq:HGising} with $|V| = 2n+3m= 18$, $|E| = n+6m=27$, $K_{is}=n+m=7$ and a constant degree $\deg(\sigma_i)=3$ for all $i$, i.e.
    \begin{equation}
      h= - \sum_{i \in V} \sigma_i + \sum_{(i,j)\in E} \sigma_i \sigma_j + \frac{23}{2} .
      \label{eq:HamPotts}
    \end{equation}
    This way we map $\phi(x_1,x_3,g_1)$, $\phi(x_2,x_4,g_4)$ and $\phi(g_1,g_2,b)$ to $h^{(1)}$, $h^{(2)}$ and $h^{(3)}$ respectively, and define
    \begin{equation}
      h_{1} = h^{(1)} + h^{(2)} + h^{(3)}
    \end{equation}
    which is defined on the graph $G_e=(V_e,E_e)$ shown in \cref{fig:Potts}.

    By construction, $\Phi$ is satisfiable iff $h_{1}$ has a ground state of 0 energy.
    Moreover, all ground states of $h_{1}$ satisfy that $\sigma_b=1$ iff $\sigma_{g_{1}}=\sigma_{g_{2}}$.
We still need to enforce that additionally $\sigma_{g_{1}}=\sigma_{g_{2}}=1$ in any ground state configuration.
We follow the same argument as in  \cref{im:PVC-Ising} in the proof of \cref{lem:2DIsing}: we multiply the hamiltonian $h_{1}$ by a large factor, say 100, and add local magnetic fields to the spins $\sigma_{g_{1}}$, $\sigma_{g_{2}}$:
\begin{equation}
h_{e,1}= 100 \: h_{1} + \frac{1-\sigma_{g_{1}} }{2} + \frac{1-\sigma_{g_{2}} }{2}  .
\end{equation}
The ground state configurations (with 0 energy) of $h_{e,1}$ must be ground state configurations of $h_{1} $,  \emph{and} $\sigma_{g_{1}}$ and $\sigma_{g_{2}}$ must be in state 1.

  \item
    We now implement Step 2 of the proof of \cref{thm:universal_characterisation}.
    This is trivial with the Ising model with fields, since it has local magnetic fields:
    \begin{equation}
      h_{e,2} = J_{1,2}\: \frac{\sigma_{b}+1}{2} .
    \end{equation}

    The Hamiltonian
    \begin{equation}
      h_e = \Delta h_{e,1}+h_{e,2}
    \end{equation}
    simulates \eqref{eq:potts-1term} as far as \cref{part:eigenvalues} and \cref{part:eigenstates} of \cref{def:simulation} are concerned.
Note that here the set of variables $P$ (see \cref{part:eigenstates} of \cref{def:simulation}) is $\sigma_P=(\sigma_1,\sigma_2,\sigma_3,\sigma_4)$.

    Finally, the Hamiltonian
    \begin{equation}
      H_G = \sum_{e\in E'} h_e
      \label{eq:IsingPotts}
    \end{equation}
    simulates the Potts Hamiltonian \eqref{eq:PottsG'} since the Ising model with fields is trivially closed.

  \item
    \label[step]{im:Z-potts}
    For the partition function, we do not require Steps 1'-3' in the proof of \cref{thm:universal_characterisation}, as the construction we have presented introduces a constant degeneracy in the number of solutions (it is an example of a $c$-monious reduction; see \cref{rem:parsimonious}).
    Explicitly, every satisfying assignment of $\phi$ (Eq.~\eqref{eq:phi}) corresponds to 3 vertex covers of size 11.
    For example, take the assignment $(u,v,w) = (1,1,1)$ and its corresponding graph.
    The triangle corresponding to the first clause has 3 valid vertex covers of size 2, whereas each of the other 3 triangles has exactly one valid vertex cover of size 2.
    An analogous argument can be made for the other satisfying assignments.
    Then the vertex covers are mapped one-to-one to the ground state configurations of the Ising model.
    Since a single edge $e$ involves 3 clauses $\phi$, we have
    \begin{equation}
      \sum_{\sigma'_1,\sigma'_2} e^{-\beta h_e'(\sigma'_1,\sigma'_2)} =
      \frac{1}{27} \sum_{ \sigma:  \:
      h_e(\sigma) <\Delta} e^{-\beta h_e(\sigma)} +
      \sum_{\sigma : \: h_e(\sigma )\geq \Delta} e^{-\beta h_e(\sigma)} .
    \end{equation}
    For the entire model we thus have
    \begin{equation}
      \sum_{\sigma'} e^{-\beta H'_{G'}(\sigma')} = \frac{1}{27^{|E'|}} \sum_{\sigma : \: H_{G}(\sigma) <\Delta}e^{-\beta H_G(\sigma)} + \sum_{\sigma : \:H_G(\sigma) \geq \Delta} e^{-\beta H_G(\sigma)} .
    \end{equation}
That is, $\gamma$ (\cref{part:partition_function} in \cref{def:simulation}) is here $1/27^{|E'|}$ and $\delta=0$.

  \end{enumerate}
\end{example}

Note that if the Potts model of \cref{eq:PottsG'} has some symmetries (e.g.~it is defined on a regular lattice, or is translationally invariant), these are recovered at a coarse-grained level in the Ising model; broadly speaking, at the scale of a the graph $G_\Phi$, which reproduces a single edge.
Also, clearly, if the Potts model has uniform couplings $J_{i,j}=J$ for all $i,j$, all flag spins will be subject to the same local fields $\Delta$.

We now consider another generalisation of the Ising model.
\begin{definition}[Clock model]
  \label{def:clock}
  The $q-$level clock model on an arbitrary graph $G=(V,E)$ with inhomogeneous couplings $J_{i,j}\in \mathbb{R}$ is defined as the set of Hamiltonians of the form
  \begin{equation}
    H_\textrm{clock,G}= \sum_{(i,j)\in E} J_{i,j}\: \cos\left[\frac{2\pi}{q} (\sigma_i-\sigma_j)\right]
  \end{equation}
  where $\sigma_i\in\{1,\ldots,q\}$.
\end{definition}

\begin{example}[Clock model to Ising model with fields]
  \label{ex:clock}
  We consider the clock model (\cref{def:clock}) with $q=4$ defined on a graph $G'=(V',E')$, and simulate it
  with the Ising model with fields (\cref{def:Isingmodel}) on a graph $G=(V,E)$.
  We proceed analogously to \cref{ex:Potts}.
  We consider one interaction
  \begin{equation}
    \label{eq:clock-1term}
    J_{1,2}\: \cos\left[\frac{2\pi}{4}(\sigma'_1-\sigma'_2)\right]
  \end{equation}
  with $\sigma'_i\in\{1,2,3,4\}$.
  We map $\sigma'_{1}$ to $(x_1,x_2)$ and $\sigma'_{2}$ to $(x_3,x_4)$, where $x_i\in\{0,1\}$.

  \begin{enumerate}
  \item
    We need only 4 flag spins, each signalling whether $\sigma'_{i}-\sigma'_{j}=\ell \mod 4$, with $\ell\in\{0,1,2, 3\}$.
    To this end, we consider $\phi$ defined in \cref{eq:phi} whose satisfying assignments are such that $w=1$ iff $u=v$.
    We define $\psi$ as the formula whose satisfying assignments are such that $w=1$ iff $u=\bar v$,
    \begin{equation}
      \psi(u,v,w)= (u \lor v\lor \bar w) \land (u \lor \bar v\lor w) \land (\bar u \lor v\lor w) \land (\bar u \lor \bar v\lor \bar w) .
      \label{eq:psi}
    \end{equation}
    Now define the boolean functions
    \begin{equation}
      \begin{array}{l}
        \Phi_{0} = \phi(x_1,x_3,u_{1}) \land
        \phi(x_2,x_4,u_{2}) \land
        \phi(u_{1},u_{2},b_0) \\
        \Phi_{1} = \phi(x_1,x_3,v_1) \land \phi(g_1,x_2,v_3) \land \psi(x_2,x_4,v_2) \land \phi(v_2,v_3,b_1) \\
        \Phi_{2} = \psi(x_1,x_3,w_1) \land \phi(x_2,x_4,w_2) \land \phi(w_1,w_2,b_2) \\
        \Phi_{3}= \phi(x_1,x_3,z_1) \land \psi(g_1,x_2,z_3) \land \psi(x_2,x_4,z_2) \land \phi(z_2,z_3,b_3).
      \end{array}
      \label{eq:phi-clock}
    \end{equation}
    One can easily verify that the satisfying assignments of $\Phi_{\ell}$ are nearly what we want, namely that $b_\ell=1$ iff $\sigma'_{1}-\sigma'_{2}=\ell \mod 4$, except for the assignments $b_{0}=1$ if $u_{1}=u_{2}=0$,
    $b_{1}=1$ if $v_{2}=v_{3}=0$,
$b_{2}=1$ if $w_{1}=w_{2}=0$,
$b_{3}=1$ if $z_{2}=z_{3}=0$.
This situation is completely analogous to that of  \cref{im:potts-formula} of \cref{ex:Potts}.
We will thus increase the energy of these assignments in the next step.

  \item
    As in \cref{im:sat-ising-potts} in \cref{ex:Potts}, the satisfiability of $\Phi_0$ is mapped to the ground state energy problem of the Ising Hamiltonian $h_{1}^{(0)}$.
    We multiply this hamiltonian by a large constant and  penalise the configurations with $u_{1}=u_{2}=0$, resulting in
    \begin{equation}
 h_{0,1}   = 100 \: h_{1}^{(0)} + \frac{1-\sigma_{u_{1}}}{2} + \frac{1-\sigma_{u_{2}}}{2} .
 \label{eq:hell1}
    \end{equation}
    We proceed similarly for $\Phi_\ell$, with $\ell=1,2,3$:
    we map its satisfiability to the ground state energy problem of
    $h_{1}^{(\ell)}$, and then obtain $ h_{\ell,1} $ as in \eqref{eq:hell1} (but substituting $u_{1}$, $u_{2}$ by $v_{2},v_{3}$ for $\ell=1$, $w_{1},w_{2}$ for $\ell=2$ and $z_{2},z_{3}$ for $\ell=3$).

    Finally, we only need add the local magnetic fields to the flag spins with the appropriate energy,
    \begin{equation}
      h_e = \Delta \sum_{\ell=1}^4 h_{\ell,1} + \sum_{\ell=1}^4 J_{1,2}\: \cos\left(\frac{2\pi}{4} \ell\right)\: \frac{\sigma_{b_{\ell}}+ 1}{2},
      \label{eq:he-potts}
    \end{equation}
in order to reproduce the energies and spin configurations of \cref{eq:clock-1term}. The overall Ising Hamiltonian
    \begin{equation}
      H_G = \sum_{e\in E'} h_e
    \end{equation}
    will reproduce the energies and configurations of the clock model on $G'$.

    \begin{figure}[t]
      \centering \includegraphics[width=0.99\textwidth]{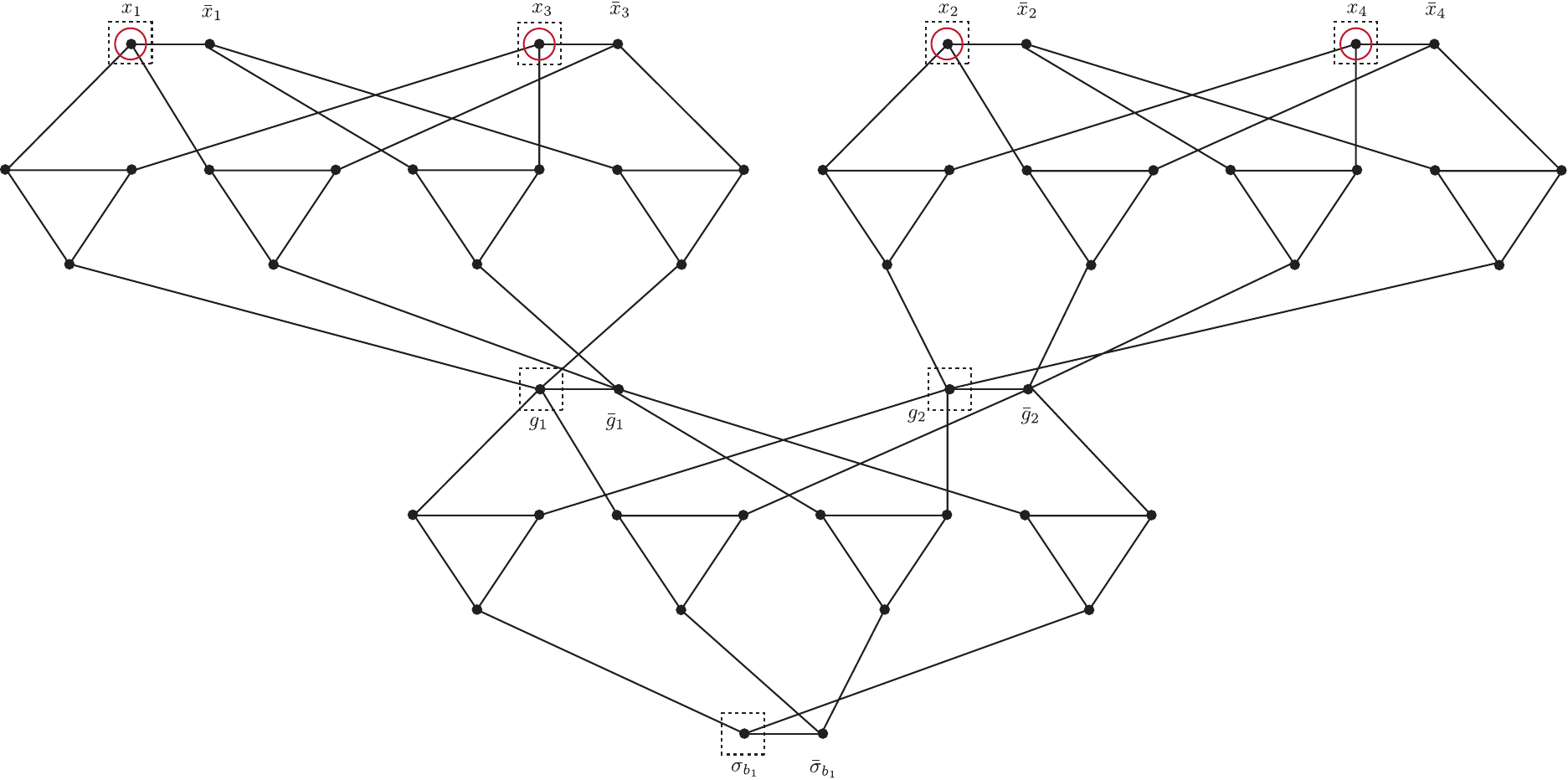}
      \caption{ The formula $\Phi$ (\cref{eq:Phi}) is satisfiable iff Ising model with Hamiltonian \cref{eq:HamPotts} in the graph $G_\Phi$ shown here has a ground state of 0 energy.
              The variables belonging to $R$ (see \cref{def:faithful}) are marked with dashed squares, and those in the physical set $P$ (see \cref{def:simulation}) are marked with red circles.
   }
      \label{fig:Potts}
    \end{figure}

  \item
    If we additionally want to reproduce the partition function, we need to make the following modification.
    Note first that each formula $\psi$ introduces a degeneracy 3 in the number of solutions (by the same argument as for $\phi$, see \cref{im:Z-potts} in \cref{ex:Potts}).
    Since the formulas of \cref{eq:phi-clock} have a different number of $\phi$ and $\psi$s, we modify $\Phi_0$ and $\Phi_2$ they will introduce a different degeneracy. This is easily solved by modifying  $\Phi_0$ and  $\Phi_2$ in an innocuous way, for example as follows,
    \begin{equation}
      \begin{array}{l}
        \Phi_0:= \phi(x_1,x_3,g_1) \land \phi(x_2,x_4,g_2) \land \phi(g_1,g_2,b_0) \land \phi(x_1,x_3,g_3)\\
        \Phi_{2} = \psi(x_1,x_3,g_1) \land \psi(x_2,x_4,g_2) \land \phi(g_1,g_2,b_2) \land \phi(x_1,x_3,g_3)
      \end{array}
    \end{equation}
    Now, each configuration $(x_1,x_2,x_3,x_4)$ introduces a degeneracy $3^4=81$.
    Defining $h_e$ as in \cref{eq:he-potts}, we obtain
    \begin{equation}
      \sum_{\sigma'_1,\sigma'_2} e^{-\beta h_e'(\sigma'_1,\sigma'_2)} = \frac{1}{81} \sum_{ \sigma, \: h_e(\sigma <\Delta} e^{-\beta h_e(\sigma)} + \sum_{\sigma_i, h_e(\sigma_i \geq \Delta} e^{-\beta h_e(\{\sigma_i\})}
    \end{equation}
    Finally, for the entire partition function we have
    \begin{equation}
      \sum_{\sigma'} e^{-\beta H'_{G'}(\sigma')} = \frac{1}{81^{|E'|}} \sum_{\sigma_i, H(\sigma) <\Delta}e^{-\beta H_G(\sigma)} + \sum_{\sigma, H_G(\sigma) \geq \Delta} e^{-\beta H_G(\sigma)}
    \end{equation}
  \end{enumerate}
\end{example}

Finally, we give an example that illustrates how to apply \cref{cor:cont} to continuous spins.

\begin{definition}[XY model]
  \label{def:xy}
  The ``XY model'' is defined as the set of Hamiltonians defined on a graph $G=(V,E)$ of the form
  \begin{equation}
    H_{\textrm{XY},G}(\{\theta_i\}) = \sum_{(i,j)\in E} J_{i,j} \cos(\theta_i-\theta_j)
  \end{equation}
  where $\theta_i\in [0,2\pi)$ for all $i\in V$.
\end{definition}
Note that the XY model can be obtained as a clock model in the limit $q\to \infty$.

\begin{example}[XY model to Ising model with fields]
  We simulate the XY model (\cref{def:xy}) on an arbitrary graph $G'$ with the Ising model with fields on another graph $G$.
  To this end, we first approximate (in the sense of \cref{def:continuous_simulation}) the XY model with a $q$-level clock model (\cref{def:clock}) on the same graph $G'$ and with the same set of couplings $\{J_{i,j}\}$.
  Finally, we simulate this clock model with an Ising model with fields as explained in \cref{ex:clock}.

  To simulate the XY model with a clock model, observe that we just need to approximate $\cos\left(\theta_i-\theta_j\right)$ by $\cos\left(2 \pi (\sigma_i-\sigma_j)/q)\right)$ with $\sigma_i\in\{1,\ldots, q\}$.
  Applying \cref{lem:continuous_term} with $L=1$, $D=2$ we achieve this with accuracy $\delta=\tfrac{2}{\sqrt{q}-1}$.
\end{example}

\end{document}